\documentclass[10pt, conference, letterpaper]{IEEEtran}
\IEEEoverridecommandlockouts
\usepackage[nocomma]{optidef}
\usepackage{microtype}
\usepackage{cite}
\usepackage{amsmath,amssymb,amsfonts}
\usepackage{graphicx}
\usepackage{textcomp}
\usepackage{xcolor}
\usepackage{booktabs} 
\usepackage{amsthm}
\usepackage{textcomp}
\usepackage{subfigure}

\usepackage[utf8]{inputenc}
\usepackage[english]{babel}
\usepackage{url}
\usepackage{algorithm}
\usepackage{algpseudocode}
\usepackage{bm}
\usepackage{comment}

\newtheorem{definition}{Definition}
\newtheorem{theorem}{Theorem}
\newtheorem{lemma}{Lemma}
\graphicspath{{figures/}} 
\DeclareMathOperator*{\argmax}{argmax}

\newcommand{\high}[1]{{\color{black}{#1}}}
\newcommand{\High}[1]{{\color{black}{#1}}}

\def\BibTeX{{\rm B\kern-.05em{\sc i\kern-.025em b}\kern-.08em
    T\kern-.1667em\lower.7ex\hbox{E}\kern-.125emX}}
\def\y{\mathbf{y}}
\def\x{\mathbf{x}}
\def\z{\mathbf{Z}}
\def\r{\mathbf{r}}
\def\q{\mathbf{q}}
\def\N{\mathcal{N}}
\def\E{\mathbb{E}}

\algrenewcommand\algorithmicrequire{\textbf{Input:}}
\algrenewcommand\algorithmicensure{\textbf{Output:}}

\begin{document}

\title{
\high{
Federated Learning with Fair Worker Selection: A
Multi-Round Submodular Maximization Approach} 
%
}
\author{
Fengjiao Li,
Jia Liu,
and Bo Ji 
\thanks{
Fengjiao Li (fengjiaoli@vt.edu) and Bo Ji (boji@vt.edu) are with the Dept. of Computer Science, Virginia Tech, Blacksburg, VA, USA. 
Jia Liu (liu@ece.osu.edu) is with the Dept. of Electrical and Computer Engineering,
The Ohio State University, Columbus, OH, USA.}
\thanks{{This work was supported in part by the NSF under Grants CNS-2112694, CNS-2110259, CCF-2110252, and a Google Faculty Research Award.}}
}


\maketitle

\begin{abstract}
In this paper, we study the problem of fair worker selection in Federated Learning systems, where fairness serves as an incentive mechanism that encourages more workers to participate in the federation. Considering the achieved training accuracy of the global model as the utility of the selected workers, which is typically a monotone submodular function, we formulate the worker selection problem
as a new multi-round monotone submodular maximization problem with cardinality and fairness constraints. The objective is to maximize the time-average utility over multiple rounds subject to an additional fairness requirement that each worker must be selected for a certain fraction of time. While the traditional submodular maximization with a cardinality constraint is already a well-known NP-Hard problem, the fairness constraint in the multi-round setting adds an extra layer of difficulty. To address this novel challenge, we propose three algorithms: Fair Continuous Greedy (FairCG1 and FairCG2) and Fair Discrete Greedy (FairDG), all of which satisfy the fairness requirement whenever feasible. Moreover, 
we prove nontrivial lower bounds on the achieved time-average utility under FairCG1 and FairCG2. In addition, by giving a higher priority to fairness, FairDG ensures a stronger short-term fairness guarantee, which holds in every round. Finally, we perform extensive simulations to verify the effectiveness of the proposed algorithms in terms of the time-average utility and fairness satisfaction.
\end{abstract}
\section{Introduction}

In the conventional machine learning paradigm, a large amount of data is often stored at a centralized server (e.g., a single machine or a datacenter) for training some learning model (e.g., a deep neural network)~\cite{googleAIBlog}. 
However, not only is this paradigm expensive in terms of data collection and storage, it also has a high risk in leaking users' data privacy~\cite{li2019survey}. For example, in 2019, hundreds of millions of Facebook user records were compromised on Amazon cloud servers~\cite{databreach}. To address such privacy concerns, Federated Learning (FL) has recently become a popular learning paradigm, where a myriad of 
worker devices collaboratively learn a global model, while all  training data are kept on the worker devices~\cite{googleAIBlog}.

In a FL system, worker devices often have heterogeneous computation or sensing capabilities. 
Due to limited resources (e.g., communication bandwidth)~\cite{lim2020federated,8851249}, it is not only desirable but also necessary for the FL server to pick a subset of workers to participate for a training task. 
Different worker selections yield different training accuracy levels that can be expressed as the utility of the selected workers. Note that the accuracy of a model has diminishing improvement as a function of sample size \cite{meek2002learning}.
%
\High{Considering the heterogeneity of the worker devices (e.g., different dataset sizes) and the redundancy among workers, the utility of the participating workers exhibits a diminishing returns property, which can be modeled as a submodular function \cite{balakrishnana2021fed}.}
In such scenarios, one needs to select a subset of workers for each training task to maximize the average utility over all tasks. On the other hand, ensuring fairness among the participating workers is one of the most important issues in such FL systems \cite{lim2020federated,li2019fair}. 
One common fairness requirement is that each worker must be selected at least for a certain fraction of time in the long run. %
Introducing such fairness requirements reduces the discrimination towards those workers with weaker sensing and computing capabilities,  thus providing an incentive that encourages more workers to participate in FL. 
Hence, there is a compelling need that motivates the worker selection problem in an FL system such that the total utility is maximized while some fairness guarantee among the workers is also ensured.

In this paper, we call the entire procedure of completing a training task from selecting participants to the end of training process a ``round". 
Noting that the training accuracy is typically a monotone submodular function of the selected workers,
we formulate the FL worker selection problem as a \underline{m}ulti-round \underline{m}onotone \underline{s}ubmodular \underline{m}aximization problem with \underline{c}ardinality and \underline{f}airness constraints, called MMSM-CF. Different from the traditional single-round version of the problem that is focused on one-shot optimization, the goal here is to maximize the time-average utility
subject to an additional fairness requirement that each worker must be selected for a minimum fraction of time. Our main contributions are summarized as follows:

\begin{list}{\labelitemi}{\leftmargin=1em \itemindent=-0.5em \itemsep=.2em}
	\item We study the problem of fair worker selection in FL systems and formulate it as a problem of  multi-round monotone submodular maximization with cardinality and fairness constraints. To the best of our knowledge, this is the first study that considers submodular maximization in the multi-round setting with fairness constraints. 
	It is well known that the single-round version of this problem even without fairness constraint is already NP-hard. Accounting for the fairness constraint in the multi-round setting adds an extra layer of difficulty as decisions could be coupled in a sophisticated fashion over multiple rounds. 
	
	\item To address this new challenge, we consider the multilinear extension of the submodular function and develop continuous greedy-based algorithms. At the first glance, it is unclear whether the original continuous greedy algorithm can achieve a provable performance guarantee in the setting with an additional 
	fairness requirement. In particular, the initial and intermediate points of the iterative continuous greedy algorithm could even be infeasible.
	Interestingly, we show that a straightforward variant of the original continuous greedy algorithm, called FairCG1, does achieve an approximation ratio of $(1-1/e)$ for the time-average utility while satisfying the fairness requirement whenever feasible. 
	Furthermore, we develop a new variant of the continuous greedy algorithm, called FairCG2, which explicitly accounts for the feasibility of the intermediate points during the updating process and achieves a fine-grained lower bound on the time-average utility. In addition, we propose a new discrete greedy algorithm, called FairDG, which gives a higher priority to fairness and thus ensures a stronger short-term fairness guarantee that holds in every round. FairDG also enjoys a much lower complexity compared to FairCG1 and FairCG2.
	
	\item Finally, we perform extensive simulations to verify the effectiveness of the proposed algorithms. The simulation results show that our proposed algorithms empirically achieve a near-optimal time-average utility (\High{within $1\%$} of the optimal). Interestingly, while FairCG1 guarantees an approximation ratio of $(1-1/e)$ uniformly for any fairness requirement, the obtained lower bound for FairCG2 is getting tighter as the fairness requirement becomes stronger.
\end{list}

The rest of the paper is organized as follows. We first discuss related work in Section~\ref{sec:related_work} and formulate the MMSM-CF problem in Section~\ref{sec:sys_prob}. To address this new problem with the long-term fairness constraint, we consider two fair continuous greedy algorithms (FairCG1 and FairCG2) and analyze their performance in Section~\ref{sec:faircg}. In Section~\ref{sec:fair_disgreedy}, we propose a fair discrete greedy algorithm (FairDG) and show that FairDG ensures a stronger short-term fairness guarantee.
Finally, we present simulation results in Section~\ref{sec:simulation} and make concluding remarks in Section~\ref{sec:conclusion}.

\section{Related Work}\label{sec:related_work}

\textbf{Worker Selection in FL:} Worker selection is an important factor of FL systems and has recently been studied in \cite{nishio2019client, 8851249}. The work of \cite{nishio2019client} proposed an FL protocol called FedCS, which selects the maximum possible number of workers using a greedy algorithm based on wireless states and devices' computing capabilities. However, fairness is not considered in this work.
In \cite{8851249}, an analytical model is developed to characterize the performance of FL with different worker scheduling policies in terms of the convergence rate performance. 
While proportional fairness is considered, scheduling decisions are made without accounting for the training accuracy. 
In \cite{li2019fair,mohri2019agnostic}, the authors considered fair resource allocation in FL systems by assigning a higher weight to workers with a higher loss. This helps reduce the bias and leads to a lower variance in the testing accuracy. 
However, the problem  formulated therein is not submodular-based and the fairness criterion they considered is very different from ours.

\smallskip
\textbf{Submodular Maximization:} Since the seminal work in \cite{nemhauser1978analysis}, the problem of monotone submodular maximization with a cardinality constraint has been extensively studied in the literature (see, e.g., \cite{krause2014submodular,buchbinder2017submodular}). It is well known that this problem is NP-hard \cite{feige1998threshold}. One important approximate solution is the classic (discrete) greedy algorithm, which iteratively selects an element with the largest marginal gain till the cardinality constraint is violated. It is shown that this greedy algorithm can achieve the best approximation ratio of $(1-1/e)$ \cite{nemhauser1978analysis}.
For submodular maximization with a general matroid constraint, it has been shown that a continuous greedy algorithm combined with the pipage rounding technique can achieve the same $(1-1/e)$-approximation \cite{vondrak2008optimal}. 
In \cite{badanidiyuru2014fast}, by integrating the idea of the classic discrete greedy, the authors developed a fast variant of the continuous greedy algorithm, which can significantly reduce the complexity while matching the same approximation ratio.
The work of \cite{chekuri2009randomized} studied a more general setting of maximizing a submodular function subject to a matroid and a set of linear packing constraints.

\smallskip
\textbf{Multi-round Submodular Maximization:} While most of the existing work has been focused on one-shot optimization, some recent studies investigate multi-round monotone submodular maximization with a cardinality constraint as we consider in this paper.  The most relevant work to ours is~\cite{sun2018multi}. However, there are several key differences: i) their work studied a specific multi-round influence maximization in social networks; ii) fairness requirements are not considered there; iii) a set would not be selected more than once because doing so does not generate any additional gain in their model. 
The work of \cite{lei2015online} 
also studied a similar problem of multi-round submodular maximization, with a focus on the online learning setting. Sequential submodular maximization has been considered in the active learning setting where a batch of data points are selected for labeling \cite{wei2015submodularity}. 
However, same data points would not be selected more than once either. Note that none of these studies addresses the same fairness concerns as ours.
Incorporating fairness constraints 
makes decisions coupled in a sophisticated fashion over multiple rounds. 
In addition, for FL systems, it is possible to select the same subset more than one round to maximize the average utility in our  problem.

\smallskip
\textbf{Fair Resource Allocation:} Algorithmic fairness has attracted tremendous interests in the machine learning community over the past few years in various contexts \cite{chouldechova2018frontiers,dwork2012fairness}. Various fairness criteria and learning paradigms have been discussed in the literature. 
Submodular maximization concerning privacy protections and fairness criteria in learning representation has been studied in~\cite{kazemi2018scalable}. However, they do not consider the multi-round setting and the fairness criterion is formulated as a robustness constraint, which is very different from ours.
The work of \cite{li2019combinatorial} studies fair resource allocation and learning in a multi-armed bandit setting. While the long-term fairness requirement we consider is the same, the objective function considered in \cite{li2019combinatorial} is linear rather than submodular.


\section{System Model and Problem Formulation} \label{sec:sys_prob}
In this section, we describe the system model and formulate the MMSM-CF problem for worker selection in FL systems. We begin with some basic notations: boldface symbols (e.g., $\x$) denote vectors; regular font symbols with subscript $u$ (e.g., $x_u$) denote the coordinate corresponding to element $u$; $\x^\mathsf{T}$ denotes the transpose of $\x$; $\x\vee\y$ denotes the coordinate-wise maximimum of $\x$ and $\y$;  $\mathbf{1}$ denotes the all-ones vector; $\mathbf{0}$ denotes the all-zeros vector; $\mathbf{1}_{u}$ denotes the standard basis vector whose coordinates are all zero, except the one corresponding to element $u$ being $1$; $\mathbb{R}$ (resp., $\mathbb{R}_+$) is the set of (resp., nonnegative) real numbers.

In the context of FL, we use $\N$ to denote the set of workers that are equipped with computing 
devices and are willing to participate in the training tasks. 
Let $n=|\N|$.  
For each training task, indexed by $t$, the FL protocol selects a subset of workers $S_t \subseteq \N$ whose size is at most $k \in \{1,2,\dots,n\}$. This cardinality constraint is used to model limited resources (e.g., communication bandwidth). After the entire training process of task $t$ is completed, a certain utility (which represents the training accuracy), denoted by $f_t(S_t)$, is achieved. Assume that these tasks are of the same type and have the same utility function, i.e., $f(\cdot)\triangleq f_t(\cdot), \forall t$, and that $f(\cdot)$ is a monotone submodular function\footnote{Consider a ground set $\N$. A set function $f: 2^{\N}\to \mathbb{R}_+$ is submodular if $f(A \cup \{u\}) - f(A) \geq f(B \cup \{u\}) - f(B)$ for every $A \subseteq B \subseteq \N$ and for every $u \in \N \setminus B$. Function $f$ is monotone if $f(A) \leq f(B)$ for every $A \subseteq B \subseteq \N$. We assume that $f(\cdot)$ is bounded and $f(\emptyset)=0$.}. Let $f(S)$ be the utility of the selected workers $S$ and $\N_k$ be the collection of all subsets of $\N$ of size at most $k$, i.e., $\N_k \triangleq \{S\subseteq \N:|S|\leq k\}$. We require $S_t\in \N_k$ for all $t$. Over a sequence of $T$ training tasks, we choose a sequence of worker sets $(S_1,S_2,\dots, S_T)$ to engage in the training tasks $\{1,2,\cdots, T\}$ and receive an average utility of $\frac{1}{T}\sum_{t=1}^T f(S_t)$. Throughout the rest of this paper, we simply use round $t$ to represent the entire training process of training task $t$ and call each worker an element of the ground set $\N$.   

Furthermore, we consider a fairness criterion of a minimum selection fraction for each individual worker and define a (long-term) fairness requirement in the following form:
\begin{equation}
	\liminf_{T\to \infty}\frac{1}{T} \sum_{t=1}^{T} \E\left[\mathbb{I}_{\{u \in S_t\}}\right]
	\geq r_u,  \, \, \forall u \in \N, \label{eq:fraction requirement}
\end{equation}
where $\mathbb{I}_{\{\cdot\}}$ is the indicator function and $r_u \in [0,1]$ is the minimum element $u$ has to be selected.
Vector $\mathbf{r}=[r_{u}]_{u\in \N}$ 
is said to be \emph{feasible} if there exists an algorithm that selects a sequence of sets $\mathcal{S} \triangleq (S_1,S_2,\dots)$ such that Eq.~\eqref{eq:fraction requirement} is satisfied. 
Given a feasible $\r$, our goal is to schedule a sequence of sets $\mathcal{S}$ that maximizes the average expected utility while satisfying the fairness requirement in Eq.~\eqref{eq:fraction requirement}. This leads to the following problem of \underline{m}ulti-round \underline{m}onotone \underline{s}ubmdoular \underline{m}aximization with \underline{c}ardinality and \underline{f}airness constraints (MMSM-CF):
\begin{maxi!}[2]
	{\substack{\mathcal{S}}} {\liminf_{T\to \infty}\frac{1}{T}\sum_{t=1}^T \E\left[f(S_t)\right]\label{eq:mmsm_cf_obj}}
	{\label{eq:MMSM-CF_problem}}{}
	\addConstraint{\eqref{eq:fraction requirement}~\text{and}~ S_t\in \N_k, \forall t \in \{1,2,\dots\}},  \label{eq:mmsm_cf_constraints}
\end{maxi!}
where the expectation is taken over all possible randomness of the considered algorithms. Assume a feasible fairness requirement $\r$. Let $U_{\mathrm{opt}}$ be the supremum value of the utility metric \eqref{eq:mmsm_cf_obj} over all feasible algorithms. Note that $U_{\mathrm{opt}}$ varies with different fairness requirements $\r$. When $\r=\mathbf{0}$, we have $U_{\mathrm{opt}}=\mathrm{OPT}$, where $\mathrm{OPT} \triangleq \max_{S\in \N_k} f(S)$ is \High{the highest utility associated with any feasible worker set, i.e., }
the optimal value for the single-round version of the problem. 

It is not difficult to see that the MMSM-CF problem is NP-hard. 
Consider the special case with fairness requirement $\r=\mathbf{0}$.
Then, the MMSM-CF problem degenerates to finding a subset of size at most $k$ that maximizes the submodular utility function in each round.
In this case, the problem 
is exactly the classic one-shot submodular maximization with a cardinality constraint, which is a well-known NP-hard problem \cite{feige1998threshold}. This simply implies that the MMSM-CF problem is also NP-hard.

Note that the fairness defined in Eq.~\eqref{eq:fraction requirement} represents a long-term requirement without any short-term guarantees. We can further strengthen the constraint to other forms of short-term fairness guarantees that hold uniformly over time. For example, consider the following short-term fairness requirement
\cite{patil2019stochastic}:
\begin{equation}
 \frac{1}{T}\sum_{t=1}^{T} \mathbb{I}_{\{u\in S_{t}\}} \geq r_u- \frac{1}{T^{\alpha}}, \forall u\in \N, \forall T \in \{1,2,\dots\}, \label{eq:alpha_fair_requirement}   
\end{equation}
where $\alpha>0$. We call an algorithm $\alpha$-fair if it satisfies Eq.~\eqref{eq:alpha_fair_requirement}. 

Clearly, by taking expectation of both sides of Eq.~\eqref{eq:alpha_fair_requirement} and letting $T$ go to infinity, the short-term fairness requirement with any given $\alpha>0$ implies the long-term fairness requirement in Eq.~\eqref{eq:fraction requirement}. Also, the larger the value of $\alpha$, the more stringent the short-term fairness requirement.  

\High{
We note that the MMSM-CF formulation is fairly general and finds applications not only in FL, but also in various networking and machine learning problems, including sensor scheduling in wireless sensor networks, task assignment in crowdsourcing, and data subset selection in machine learning. (See Appendix~A 
for more detailed discussions.)}
In the following, we will first focus on the design of approximation algorithms for the MMSM-CF problem with long-term fairness guarantees in Section~\ref{sec:faircg} and then design an $\alpha$-fair algorithm with short-term fairness guarantees in Section~\ref{sec:fair_disgreedy}.

\section{MMSM-CF with Long-term Fairness Guarantees}\label{sec:faircg}
In this section, we first reformulate the MMSM-CF problem with long-term fairness requirement as a Linear Programming (LP) by considering a class of stationary randomized policies. Then, we develop two continuous greedy algorithms (FairCG1 and FairCG2) and show that they both can approximately solve the MMSM-CF problem. Specifically, we show that the fairness requirement in Eq.~\eqref{eq:fraction requirement} is satisfied 
whenever feasible (Theorems~\ref{thm:fairness_guarantee} and \ref{thm:fairness_guarantee2}) and prove nontrivial lower bounds on the achieved time-average utility (Theorems~\ref{thm:approx_ratio_faircg1} and \ref{thm:approx_ratio_faircg2}).
\subsection{LP-based Reformulation}\label{sec:lp_reformulation}
Since the utility function does not change over rounds, the order of the selected worker sets $S_1,S_2,\dots$ does not impact the average utility. Hence, it suffices to find an optimal assignment of time fractions among all the sets in $\N_k$. Then, each set will be selected in the corresponding fraction of rounds.

We consider a class of \emph{stationary randomized policies} that randomly choose a set in each round. Consider such a stationary randomized policy $\pi$, which is characterized by a probability distribution $\mathbf{q}= [q_S]_{S\in \N_k}$, where $q_S$ is the probability of choosing set $S\in \N_k$. Then, we have $\sum_{S\in \N_k}q_S=1$. Under policy $\pi$, the following is satisfied for all $u \in \N$ and for all $t \in \{1,2,\dots\}$:
\begin{equation}
	\begin{aligned}
		\E\left[\mathbb{I}_{\{u\in S_t\}}\right]
		&= \E\left[ \sum_{S\in \N_k: u\in S} \mathbb{I}_{\{S_t=S\}}\right] \\
		&= \sum_{S\in \N_k: u\in S}  \E\left[\mathbb{I}_{\{S_t=S\}} \right]
		= \sum_{S\in\N_k:u\in S} q_S.\label{eq:randomized_fair}
	\end{aligned}
\end{equation}
Then, Eq.~\eqref{eq:fraction requirement} can be rewritten as $\sum_{S\in \N_k:u\in S} q_S \geq r_u$ for all $u \in \N$. 
The time-average utility in Eq.~\eqref{eq:mmsm_cf_obj} can also be rewritten as $\sum_{S\in \N_k} q_S f(S)$. Therefore, for this class of stationary randomized policies, the MMSM-CF problem can be reformulated as the following LP : 
\begin{maxi!}[2]
	{\substack{\mathbf{q}}} {\sum_{S\in \N_k} q_Sf(S) \label{eq:optimal_LP_obj}}
	{\label{eq:optimal_LP}}{}
	\addConstraint{\sum_{ S\in\N_k: u\in S } q_S \geq r_u, \quad \forall u\in \N}  \label{eq:optimal_LP_fair}
	\addConstraint{\sum_{ S\in \N_k} q_S= 1 \label{eq:optimal_LP_dist1}}
	\addConstraint{ q_S \in [0,1], \forall S\in\N_k. \label{eq:optimal_LP_dist2}}
\end{maxi!}

\begin{lemma}\label{lem:randomized_exist} Suppose that a fairness requirement $\r$ is feasible. Then, there is a stationary randomized policy that is optimal for the MMSM-CF problem in \eqref{eq:MMSM-CF_problem}. 
\end{lemma}
We omit the proof of Lemma~\ref{lem:randomized_exist}, as it follows directly from \cite[Theorem~4.5]{neely2010stochastic}, where the objective is to minimize a time-average penalty (equivalent to maximizing the time-average utility in Eq.~\eqref{eq:mmsm_cf_obj}) while keeping all the queues mean rate stable (equivalent to the fairness requirement in Eq.~\eqref{eq:fraction requirement} being satisfied). 
Due to Lemma~\ref{lem:randomized_exist}, we can focus on \High{
finding an optimal stationary randomized policy for the MMSM-CF problem.}

While an LP can be solved in polynomial time with respect to the number of variables, in the above LP problem, the number of variables is exponential in the size of the input (i.e., $n$ and $k$). Solving the above LP requires $O(n^k)$ \emph{oracle queries}\footnote{We assume access to a \emph{value oracle} 
\cite{vondrak2008optimal}.
One oracle query means the value of $f(S)$ returned by the value oracle, provided an input $S\subseteq \N$.} to obtain the value of $f(S)$ for every set $S \in \N_k$. Since the MMSM-CF problem is NP-hard, we aim to develop efficient approximation algorithms.

In the following, we present two Fair Continuous Greedy algorithms, FairCG1 and FairCG2, by carefully integrating randomized dependent rounding with the (modified) continuous greedy algorithm based on multilinear extension \cite{vondrak2008optimal}.

\subsection{FairCG1} \label{sec:algorithm}
As discussed in Subsection~\ref{sec:lp_reformulation}, solving the MMSM-CF problem is equivalent to finding a distribution that leads to a fractional vector on $[0,1]^\N$.
This motivates us to explore continuous extensions of submodular functions and continuous methods that perform on $[0,1]^\N$. The multilinear extension is an important extension of submodular functions and has unique properties that are useful for (single-round) submodular maximization subject to a matroid\footnote{A \emph{matroid} is a pair $(\N, \mathcal{I})$ such that $\N$ is a finite set, and $\mathcal{I}$ is a non-empty collection of subsets of $\N$ satisfying the following properties: a) $A\subseteq B\subseteq \N$ and $B\in \mathcal{I}$ implies $A \in \mathcal{I}$; b) for any two sets $A, B\in\mathcal{I}$, with $|A|<|B|$, there exists an element $u\in B\backslash A$ such that $A\cup \{u\}\in\mathcal{I}$. The sets in $\mathcal{I}$ are called \textit{independent sets} \cite{krause2014submodular,calinescu2007maximizing}.} constraint. 
We restate the definition of multilinear extension \cite{vondrak2008optimal} as follows.
\begin{definition}
	For a set function $f: 2^{\N} \to \mathbb{R}$, its multilinear extension $F: [0,1]^{\N} \to \mathbb{R}$ is defined as
	\begin{equation}
		F(\y) \triangleq \sum_{S\subseteq \N} f(S)\prod_{u\in S } y_u\prod_{v\notin S}(1-y_v). \label{eq:multilinear_def}
	\end{equation}
\end{definition}
From the above definition, we can see that the multilinear extension $F(\y)$ is the expectation of $f(S)$ with $S$ determined by \High{selecting each element $u$ independently with probability $y_u$.} 
\High{Consider submodular maximization subject to a general matroid. One can obtain the relaxed maximization problem by replacing the submodular function with its multilinear extension and the original matroid with its corresponding matroid polytope\footnote{For a matriod $(\N, \mathcal{I})$, the matroid polytope is the convex hull of all the characteristic vectors of the independent sets in $\mathcal{I}$. Here, the characteristic vector of a set $S$ is the $n$-dimensional vector form of $S$, where the coordinate corresponding to every element $u \in S$ is equal to $1$ and the other coordinates are all equal to $0$.}.}
Then, one can attempt to approximate the original integer problem through the following two steps: i) finding an approximate fractional solution to the relaxed problem; ii) rounding the fractional solution to an integral solution without losing too much objective value~\cite{buchbinder2017submodular}. 

The continuous greedy algorithm is an efficient method (polynomial in $n$) to approximate the relaxed problem.
It maintains a vector $\y(\tau)$ that evolves during the time interval $\tau \in [0,1]$. Specifically, it starts with a zero-vector solution, i.e., $\y(0)=\mathbf{0}$, gradually updates the vector on the coordinates with maximal improvement, and finally generates a fractional vector $\y(1)$ in the polytope. 
In \cite{vondrak2008optimal}, it is shown that the output $\y(1)$ achieves an approximation ratio of $(1-1/e)$, i.e., $F(\y(1)) \geq (1-1/e) \cdot \mathrm{OPT}$ (with a small discretization error $o(1)$). 
Furthermore, thanks to the convexity of the multilinear extension in any direction $\mathbf{d} = \mathbf{1}_u-\mathbf{1}_v$ for any two different elements $u, v \in \N$, one can perform pipage rounding \cite{ageev2004pipage} on $\y(1)$ to obtain a subset $S\subseteq\N$ that satisfies $f(S) \geq F(\y(1)) \geq (1-1/e) \cdot \mathrm{OPT}$.

Next, we leverage the properties of 
the multilinear extension and extend the continuous greedy algorithm to address the MMSM-CF problem. 
\subsubsection{Algorithm Design}
Let $\y=[y_u]_{u\in \N}\in[0,1]^\N$ 
with $y_u$ being the probability 
of selecting element $u$. Obviously, under a randomized policy, $y_u$ is the expected time fraction of selecting $u$.
To comply to the constraints of Problem~\eqref{eq:optimal_LP}, $\y$ needs to satisfy $\r \leq \y \leq \mathbf{1}$ (fairness constraint) and $ \y^\mathsf{T} \mathbf{1} \leq k$ (cardinality constraint). Let $P_{f}$ be the feasible region for 
$\y$:
\begin{equation}
	P_f \triangleq \{\y \in [0,1]^\N: \r \leq \y \leq \mathbf{1} ~\text{and}~ \y^\mathsf{T}\mathbf{1}\leq k\}. \label{eq:polytope}  
\end{equation} 
Note that $P_f$ is a polytope.
Hence, one intuitive approach is to apply continuous greedy and pipage rounding technique presented in \cite{vondrak2008optimal}, where they study submodular maximization subject to a matroid contraint. In \cite{buchbinder2017submodular}, it is shown that the continuous greedy method can be applied to arbitrary convex body with zero-vector inside and achieves an approximation ratio of $(1-1/e)$
. 
In our problem, however, the feasible region $P_f$ does not contain the origin or any of the characteristic vectors of the sets of size $k$ when $\r>0$. This results in infeasible intermediate points (i.e., outside of the feasible region $P_f$) during the iterative process. Therefore, it is unclear whether this algorithm can achieve the same approximation ratio of ($1-1/e$) or not for the MMSM-CF problem. Specifically, the following three aspects are unclear:
 
\High{\emph{i) Since the initial vector and the intermediate vectors during the iterative process of continuous greedy may be infeasible, it is unclear whether the final solution $\y(1)$ is feasible or not. }}

\emph{ii) Both the objective function and the form of the optimal solution of the MMSM-CF problem are very different from the problems studied in the literature. It is unclear whether directly applying continuous greedy on $P_f$ can still achieve an approximation ratio of $(1-1/e)$ for the MMSM-CF problem.} 

\emph{iii) A deterministic pipage rounding does not work for the MMSM-CF problem with multi-round nature since it only provides a one-shot performance guarantee for the achieved utility without respecting the fairness requirement.}

To address the new MMSM-CF problem, we develop the first fair continuous greedy algorithm, called \emph{FairCG1}, by directly employing the continuous greedy algorithm on $P_f$ in the first step. After obtaining the fractional vector, we sequentially perform randomized dependent rounding (i.e., randomized pipage rounding) on the fractional output. Interestingly, we can show that not only can FairCG1 satisfy the fairness requirement in Eq.~\eqref{eq:fraction requirement}, but it also achieves $(1-1/e)$-approximation.
\begin{algorithm}[tb]	
	\caption{Fair Continuous Greedy with Randomized Dependent Rounding 1 (FairCG1)}
	\label{alg:faircg1}
	\begin{algorithmic}[1]
	\Require $\mathbf{r}$, $k$, and $\N$
	\Ensure $\mathcal{S}=(S_1,S_2,\dots)$
		\item[] $//$\emph{Run Fair Continuous Greedy in $P_f$ and get $\y(1)$}
		\State Set $\tau=0$ and $\mathbf{y}(0)=\mathbf{0}$ \label{alg_init}  
		\While{$\tau<1$} \label{alg_while_start}
		\For{$u\in \N$}
		\State Let $w_u(\tau)=F(\y(\tau)\vee \mathbf{1}_u)-F(\y(\tau))$   \label{alg_compute_w}
		\EndFor
		\State Let $\mathbf{w}(\tau) = [w_u(\tau)]_{u \in \N}$
		\State Find $\x(\tau) = \argmax_{\x \in P_f} \x^\mathsf{T}\mathbf{w}(\tau)$ 
		\State Increase $\y(\tau)$ at a rate of $\frac{d\y(\tau)}{d\tau} = \x(\tau)$, where $\tau$ is increased by $d\tau$ \label{alg_increase}
		\EndWhile  \label{alg_while_end}
		\item[] $//$\emph{Perform randomized dependent rounding on $\y(1)$}
		\For{$t=1,2, \cdots$} \label{alg_round_start}
		\State $S_t = $\textsc{DepRounding}($\y(1)$)  \label{alg_dr}
		\EndFor \label{alg_round_end}
		\Function{DepRounding}{$\y$} \label{alg_depRd_start}
				\While{$\exists u\in \N$ with $y_u \in (0,1)$} 
			\State find $u,v,~\text{and}~u\neq v$ such that $y_u, y_v \in (0,1)$
			\State $a=\min\{1-y_u, y_v\}$, $b=\min\{y_u,1-y_v\}$
			\State 	$(y_u, y_v) =
			\begin{cases}
			(y_u+a, y_v-a), ~\text{w.p.}~ \frac{b}{a+b} \\ \label{eq: played numbers updating}
			(y_u-b, y_v+b), ~\text{w.p.}~ \frac{a}{a+b}
			\end{cases}$
			\EndWhile
			\State \textbf{return} \ $S = \{u \in \N: y_u=1\}$
		\EndFunction \label{alg_depRd_end}
	\end{algorithmic}
\end{algorithm}

The detailed operations of FairCG1 are presented in Algorithm~\ref{alg:faircg1}.
It consists of two main steps. 

\emph{Step~1}: Run the original continuous greedy algorithm in the polytope $P_f$ defined in Eq.~\eqref{eq:polytope}. FairCG1 keeps a fractional vector $\y(\tau)$ starting from $\mathbf{0}$ and evolving during the time interval $[0,1]$. At each time point $\tau$, it finds a vector $\x(\tau)$ in $P_f$ that maximizes the dot product of $\x$ and $\mathbf{w}(\tau)$ and updates $\y(\tau)$ with a rate of $\x(\tau)$. By the end of this step
, we obtain a fractional vector $\y(1)$. It is not difficult to show $\y(1)^\mathsf{T} \mathbf{1}=k$.

\emph{Step~2}: With the fractional vector $\y(1)$ obtained from Step~1, we employ the randomized dependent rounding function, \textsc{DepRounding}, and obtain a random set $S_t$ in each round $t$. Given any input $\y$ satisfying $\y^\mathsf{T} \mathbf{1}=k$, function \textsc{DepRounding} selects $k$ elements from $\N$ with a probability specified by vector $\y$. Here, we only consider set $S_t$ of size $k$ since the utility function $f(\cdot)$ is monotone.
The detailed operations of \textsc{DepRounding} are presented in Lines~\ref{alg_depRd_start}-\ref{alg_depRd_end}.

\textbf{Remark:} 
In the first step of FairCG1, we implement the fair continuous greedy method by performing discretization in two aspects as in \cite{vondrak2008optimal}: i) we increase time $\tau$ by small steps of $d\tau = 1/n^2$; ii) each $F(\y)$ is evaluated with $n^5$ independent samples of $f(R(\y))$, where $R(\y)$ is a random set where each element $u$ appears independently with probability $y_u$. The value of $f(R(\y))$ for each sample is obtained via an oracle query. More details of the discretization process can be found in \cite{vondrak2008optimal}, where they showed that the above discretization error is $o(1)$.

\subsubsection{Performance of FairCG1}
We present the main results regarding the performance of FairCG1: i) whether it satisfies the fairness constraints or not and ii) the achieved time-average utility.
First, we show that FairCG1 guarantees the long-term fairness requirement in Eq.~\eqref{eq:fraction requirement} as long as the requirement vector $\mathbf{r}$ is feasible. We state this result in Theorem \ref{thm:fairness_guarantee}. 
\begin{theorem} \label{thm:fairness_guarantee}
	For any feasible requirement $\r$, FairCG1 satisfies the fairness requirement in Eq.~\eqref{eq:fraction requirement}. Furthermore, for any finite $T$,
	the selection fraction of each element $u$ deviates from its fairness requirement $r_u$ by any $\delta>0$ with probability at most $e^{-2T\delta^2}$. That is, we have
    \begin{equation}
		\mathbb{P}\left\{\frac{1}{T}\sum_{t=1}^{T} \mathbb{I}_{\{u\in S_t\}} \leq  r_u - \delta\right\} \leq e^{-2T\delta^2}. \label{eq:results_fairness}
	\end{equation}
\end{theorem}
To prove Theorem~\ref{thm:fairness_guarantee}, we first show that the fractional vector $\y(1)$ obtained from Step~1 of FairCG1 satisfies $\y(1) \ge \mathbf{r}$. 
Then, in each round $t$, we select each element $u$ with probability $y_u(1)$ using \textsc{DepRounding}. Therefore, the fairness requirement of Eq.~\eqref{eq:fraction requirement} is satisfied. 
Furthermore, by applying the Hoeffding Bound \cite{hoeffding1994probability}, we can also show the probabilistic guarantee for fairness satisfaction in Eq.~\eqref{eq:results_fairness} for any finite $T$. 

Next, we prove a nontrivial lower bound on the time-average utility achieved by FairCG1. Recall that $U_{\mathrm{opt}}$ is the optimal value of Problem \eqref{eq:MMSM-CF_problem}. We state the lower bound in Theorem~\ref{thm:approx_ratio_faircg1}. 
\begin{theorem} \label{thm:approx_ratio_faircg1}
	The expected time-average utility under FairCG1
	has the following lower bound:
	\begin{equation}
	\liminf_{T\to \infty}	\frac{1}{T}\sum_{t=1}^{T}\E\left[f(S_t)\right] \geq (1-1/e) \cdot U_{\mathrm{opt}}. \label{eq:approx_ratio_faircg1}
	\end{equation}
\end{theorem}
To prove Theorem~\ref{thm:approx_ratio_faircg1}, we first show that the fractional vector $\y(1)$ satisfies $F(\y(1)) \geq (1-1/e) U_{\mathrm{opt}} $ and then prove $\mathbb{E}[f(S_t)]\geq F(\y(1))$ for every round $t \in \{1,2,\dots\}$. This further implies Eq.~\eqref{eq:approx_ratio_faircg1} and completes the proof.
For the first part, we follow a similar line of analysis for the continuous greedy algorithm in \cite{buchbinder2017submodular}. The difference lies in the different forms of the optimal solution and the optimal value. The optimal solution we consider is a distribution $\mathbf{q}^*$ over feasible sets in $\N_k$, while in \cite{buchbinder2017submodular}, it is a set $S^* \in \argmax_{S\in \N_k} f(S)$. Besides, the optimal value of the MMSM-CF problem is $U_{\mathrm{opt}}= \sum_{S\in \N_k} q^*_Sf(S)$ compared to $f(S^*)$ in \cite{buchbinder2017submodular}. 
An interesting insight we obtain from the analysis is the following: Despite the above difference, the continuous greedy algorithm can still achieve an approximation ratio of $(1-1/e)$ even if it is applied to a convex region that does not contain the origin or the characteristic vectors of any sets with size $k$.

The above lower bound holds uniformly for any feasible fairness requirement $\mathbf{r}$.  If $\r=\mathbf{0}$, we have $U_{\mathrm{opt}} = \mathrm{OPT}$. 
Thus, we have $\E\left[f(S_t)\right] \geq  (1-1/e) \cdot \mathrm{OPT}$, which recovers the best possible approximation ratio for the classic one-shot monotone submodular maximization with a cardinality constraint.

\subsection{FairCG2}
In this subsection, we go one step further and develop a new variant of the continuous greedy algorithm, called \emph{FairCG2}, which starts with an initial point $\y(0) =\mathbf{r}$ such that the intermediate points $\y(\tau)$ are always kept in the feasible region $P_f$ during the updating process. 
We prove that FairCG2 achieves a fine-grained lower bound on the time-average utility, which can be characterized by the fairness requirement $\mathbf{r}$.

\subsubsection{Algorithm Design} In the first step, FairCG2 starts with $\r$ (instead of $\mathbf{0}$ as in FairCG1) and updates the vector with an adjusted rate.  
We present the details of FairCG2 in Algorithm~\ref{alg:faircg2}. Similar to FairCG1, FairCG2 consists of two main steps. The differences are as follows: i) the initial point is $\r$ rather than $\mathbf{0}$ (Line~\ref{alg_init}); ii) the updating rate of $\y(\tau)$ is $\frac{d\y(\tau)}{d\tau}=\x(\tau)-\r$ instead of $\x(\tau)$ (Line~\ref{alg_increase}), where $\x(\tau)\triangleq \argmax_{\x\in P_f}\x^{\mathsf{T}}\mathbf{w}(\tau)$. To better understand the difference between FairCG1 and FairCG2, we illustrate the updating processes of $\y(\tau)$ under these two algorithms in Fig.~\ref{fig:update_process}. By starting from $\mathbf{0}$, some points $\y(\tau)$ along the path taken by FairCG1 may not be feasible (i.e., outside of $P_f$) while every point $\y(\tau)$ under FairCG2 is feasible. However, the output $\y(1)$ is feasible under both FairCG1 and FairCG2. We also have $\y(1)^\mathsf{T} \mathbf{1}=k$ under FairCG2. Then, we round the fractional vector $\y(1)$ with randomized dependent rounding method in each round. 
In addition, the discretization 
in FairCG2 is performed in the same way as that in FairCG1.
\begin{figure}
    \centering
    \subfigure[FairCG1]{
    \label{fig:faircg1}
    \includegraphics[width =0.46\linewidth]{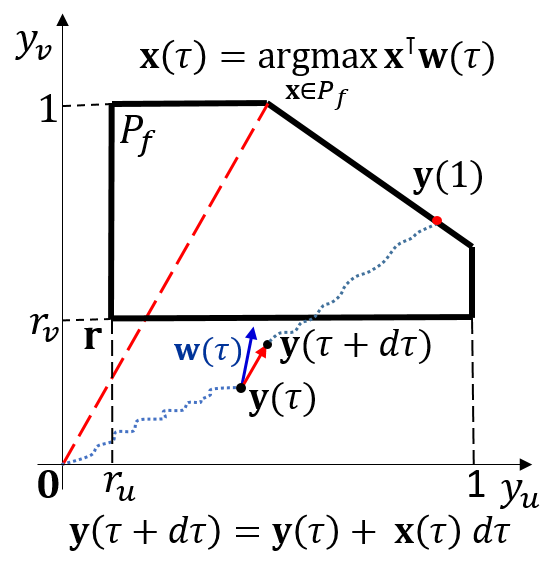}
    }
    \subfigure[FairCG2]{
    \label{fig:faircg2}
    \includegraphics[width=0.47\linewidth]{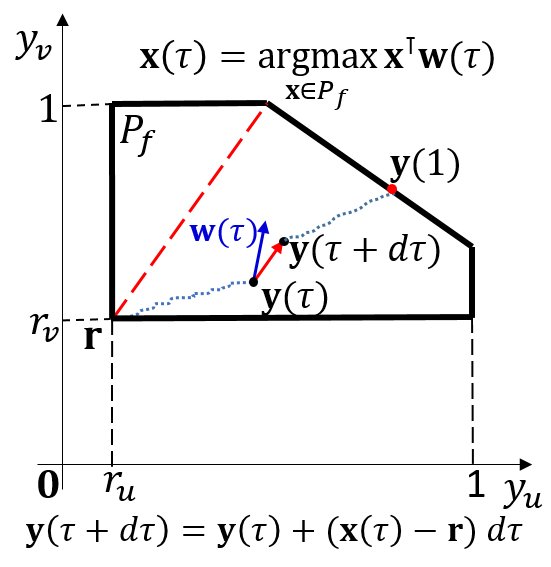}
    }
    \caption{Updating process of $\y(\tau)$ under FairCG1 and FairCG2.}
    \label{fig:update_process}
    \vspace{-1pt}
\end{figure}
\begin{algorithm}[tb]
	\caption{Fair Continuous Greedy with Randomized Dependent Rounding 2 (FairCG2)}
	\label{alg:faircg2}
	\begin{algorithmic}[1]
	\State Same as Algorithm~\ref{alg:faircg1} except for Lines~\ref{alg_init} and \ref{alg_increase}:
	\item[] Line~\ref{alg_init}: Set $\tau=0$ and $\y(0)=\r$
	 \item[]  Line~\ref{alg_increase}: Increase $\y(\tau)$ at a rate of $\frac{d\y(\tau)}{d\tau}=\x(\tau)-\r$, where $\tau$ is increased by $d\tau$
	\end{algorithmic}
\end{algorithm}

\subsubsection{Performance of FairCG2} \label{sec: performance_analysis}
We present the main results for FairCG2 in terms of fairness satisfaction and utility.

First, it is not difficult to show that FairCG2 guarantees the long-term fairness requirement in Eq.~\eqref{eq:fraction requirement} whenever the requirement $\mathbf{r}$ is feasible. We state this result in Theorem~\ref{thm:fairness_guarantee2} and omit the proof since it is almost the same as that of Theorem~\ref{thm:fairness_guarantee}. 

\begin{theorem} \label{thm:fairness_guarantee2}
	For any feasible requirement $\r$, FairCG2 satisfies the fairness requirement in Eq.~\eqref{eq:fraction requirement} \High{and also guarantees Eq.~\eqref{eq:results_fairness}.}
\end{theorem}
Next, we show that by starting from $\r$ and updating the process with rate $\x(\tau)-\r$, FairCG2 offers a fine-grained lower bound on the achieved time-average utility, which can be characterized by the fairness requirement $\r$.
We present this result in Theorem~\ref{thm:approx_ratio_faircg2}. Recall that $U_{\mathrm{opt}}$ is the optimal value of Problem~\eqref{eq:MMSM-CF_problem} and that $F(\cdot)$ is the multilinear extension of $f(\cdot)$.
\begin{theorem} \label{thm:approx_ratio_faircg2}
	The time-average expected utility under FairCG2 has the following lower bound:
	\begin{equation}
		\liminf_{T\to \infty}\frac{1}{T}\sum_{t=1}^{T}\E\left[f(S_t)\right] \geq (1-1/e^{c_{\mathbf{r}}}) \cdot U_{\mathrm{opt}} +F(\mathbf{r})/e^{c_{\mathbf{r}}}, \label{eq:approx_ratio_faircg2}
	\end{equation}
	where $c_{\mathbf{r}} \triangleq 1-\max\{\max_u r_u, \r^\mathsf{T}\mathbf{1}/k\}$.
\end{theorem}
Similar to the analysis of FairCG1, we prove Theorem \ref{thm:approx_ratio_faircg2} as follows: we first show that the fractional vector $\y(1)$ satisfies $F(\y(1)) \geq (1-1/e^{c_{\mathbf{r}}}) \cdot U_{\mathrm{opt}} +F(\mathbf{r})/e^{c_{\mathbf{r}}}$ and then prove $\mathbb{E}[f(S_t)]\geq F(\y(1))$ for each $t$. 
The analysis is slightly different since both the starting point and the updating rate depend on the fairness requirement $\r$. 

\textbf{Remark:}
The lower bound in Theorem~\ref{thm:approx_ratio_faircg2} appears to be getting tighter as the fairness requirement becomes more stringent. This can be observed in the simulation results (see Fig.~\ref{fig:impact_of_r}). Furthermore, we can show that the lower bound is at least $(1-1/e) \cdot U_{\mathrm{opt}}$ in two extreme cases: i) $\mathbf{r}=\mathbf{0}$; and ii) $\r=\r^{\prime}$ with ${\r^{\prime}}^\mathsf{T}  \mathbf{1} = k$. 
When $\r=\mathbf{0}$, we have $c_{\r} = 1$, $F(\r)=0$, and $U_{\mathrm{opt}} = \mathrm{OPT}$. 
This implies $\E\left[ f(S_t)\right] \geq  (1-1/e) \cdot \mathrm{OPT}$, which recovers the best possible approximation ratio. In another extreme case of $\r=\r^{\prime}$
, we have $c_{\r^{\prime}} = 0$, $P_f=\{\r^{\prime}\}$, and thus $\y(1)=\r^{\prime}$. According to \cite[Theorem~II.1]{chekuri2010dependent}, we obtain
$\E\left[f(S_t)\right] \geq F(\r^{\prime}) $.
Let $U_{\mathrm{opt}}{(\r^{\prime})}$ denote the optimal time-average utility with fairness requirement $\r^{\prime}$.
Combining the result in \cite{calinescu2007maximizing}, it is not difficult to derive $F(\r^{\prime})\geq (1-1/e) \cdot U_{\mathrm{opt}}(\r^{\prime})$, and then, we have $\E\left[ f(S_t)\right] \geq  (1-1/e) \cdot U_{\mathrm{opt}}$ as well for $\r=\r^{\prime}$.

\subsection{Complexity of FairCG1 and FairCG2}
To study the complexity of an algorithm for submodular  optimization, 
it is common to analyze the number of oracle queries required by the algorithm \cite{buchbinder2017submodular}. Here, we consider not only the number of oracle queries but also the running time of FairCG1 and FairCG2, since the problem itself has a multi-round nature. 
Specifically, the first step of FairCG1 and FairCG2 is a variant of the continuous greedy algorithm, implemented with discretization
, which results in $2n^8$ oracle queries in the first step.
Consider $T$ rounds. 
During the second step, performing \textsc{DepRounding$(\y(1))$} 
takes at most $n$ iterations in each round, which results in $O(nT)$ running time with no oracle queries. Therefore, the complexity of FairCG1 and FairCG2 is\footnote{We use $O^\dagger(\cdot)$ to denote the number of oracle queries.} $(O^{\dagger}(n^8)+ O(nT))$, which could be quite high as $n$ gets large. This motivates us to further develop low-complexity approximation algorithms to solve MMSM-CF.

The work of \cite{badanidiyuru2014fast} developed a fast variant of the continuous greedy algorithm that integrates the discrete greedy algorithm and has a significantly lower complexity, while matching the best known approximation ratio of $(1-1/e)$.
The basic idea is the following. In each while loop iteration (i.e., Lines~\ref{alg_while_start}-\ref{alg_while_end}) of Algorithm~\ref{alg:faircg1}, continuous greedy aims to find some $\x(\tau)$ in the polytope, which is the characteristic vector of some size-$k$ subset when the constraint is a matroid. Such a subset can be approximately found using the idea of discrete greedy with a much lower complexity. However, the fairness constraint we consider renders this fast continuous greedy algorithm inapplicable. 
This is because $\x(\tau) \in P_f$ is not necessarily a characteristic vector of some size-$k$ subset in our case. Therefore, discrete greedy cannot be applied here. It also remains unclear how one can adapt such a fast continuous greedy algorithm to address the MMSM-CF problem.

\section{MMSM-CF with Short-Term Fairness Guarantees}
\label{sec:fair_disgreedy}
In this section, by taking into account the fairness requirement, we develop a new variant of the discrete greedy algorithm, which has a much lower complexity than FairCG1 and FairCG2, especially when the size of the ground set (i.e., $n$) is large. Moreover, this new algorithm ensures the stronger short-term fairness requirement in Eq.~\eqref{eq:alpha_fair_requirement}.
\begin{algorithm}[t]
	\caption{Fair Discrete Greedy (FairDG)}
	\label{alg:discrete_alg}
	\begin{algorithmic}[1]
	\Require $\mathbf{r}$, $k$, and $\N$ 
	\Ensure $\mathcal{S}=(S_1,S_2,\dots)$
		\State Initialize $N_{u,0}=0 $ for all $u \in \N$
		\For{$t=1,2, \cdots$}
		\State $A_t = \{u\in \N: r_u t-N_{u,t-1}\geq 0\}$
		\State $l = |A_t|$ 
		\If{$l< k$}
		\State $B_0 = A_t$ 
		\For{$j = 1, \cdots, k-l$}
		\State $ v \in \argmax_{u\in \N\backslash B_{j-1}} \Delta(u|B_{j-1})$
		\State $B_{j}= B_{j-1}\cup \{v\}$
		\EndFor
		\State $S_t=B_{k-l}$
		\Else
		\State $B_0 = \emptyset$
		\For{$j = 1, \cdots, k$}
		\State $ v \in \argmax_{u\in \N\backslash B_{j-1}} (r_ut-N_{u,t-1})$
		\State $B_{j}= B_{j-1}\cup \{v\}$
		\EndFor
		\State $S_t=B_{k}$
		\EndIf
		\EndFor
	\end{algorithmic}
\end{algorithm}
Recall that the classic discrete greedy algorithm for single-round monotone submodular maximization with a cardinality constraint starts with an empty set and iteratively adds an element with the largest marginal gain until the cardinality constraint is violated. This simple and efficient greedy algorithm achieves the best possible approximation ratio of $(1-1/e)$ \cite{nemhauser1978analysis}. 
Hence, a natural idea is to give a higher priority to the fairness requirement while being greedy.
Specifically, in each round, we first check the violation of the fairness requirement for each element and then make decisions in a greedy manner by giving the unsatisfied element a higher priority. We call this new algorithm Fair Discrete Greedy (\emph{FairDG}). The detailed operations of FairDG are presented in Algorithm~\ref{alg:discrete_alg}. 
\begin{figure*}[!t]
		\centering
		\subfigure[Utility over rounds]{
			\label{fig:utility_evolve}
			\includegraphics[width=0.32\textwidth]{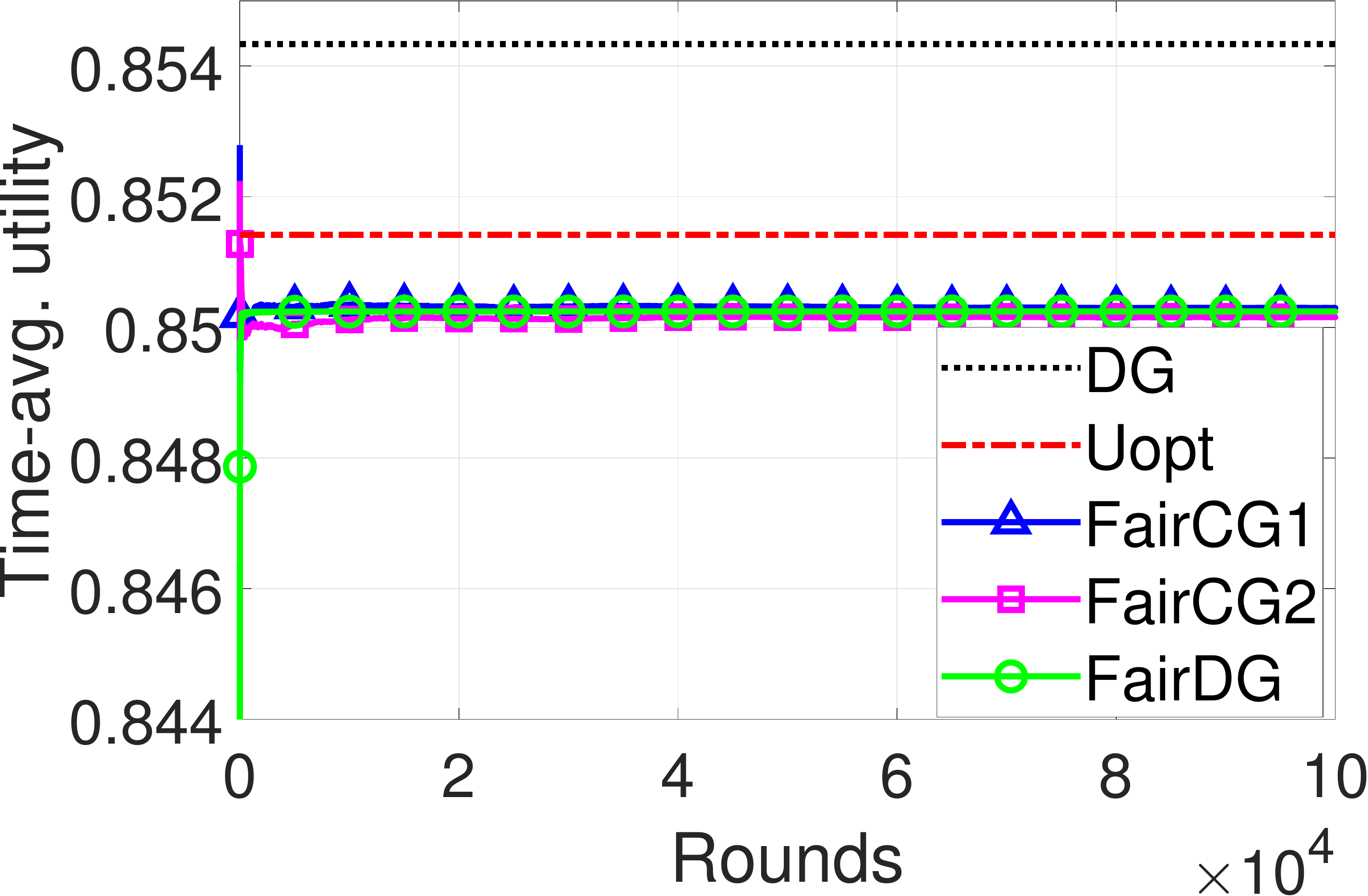}}
		\quad
		\subfigure[Selection fractions in $T$ rounds]{
			\label{fig:fraction} 
			\includegraphics[width=0.3\textwidth]{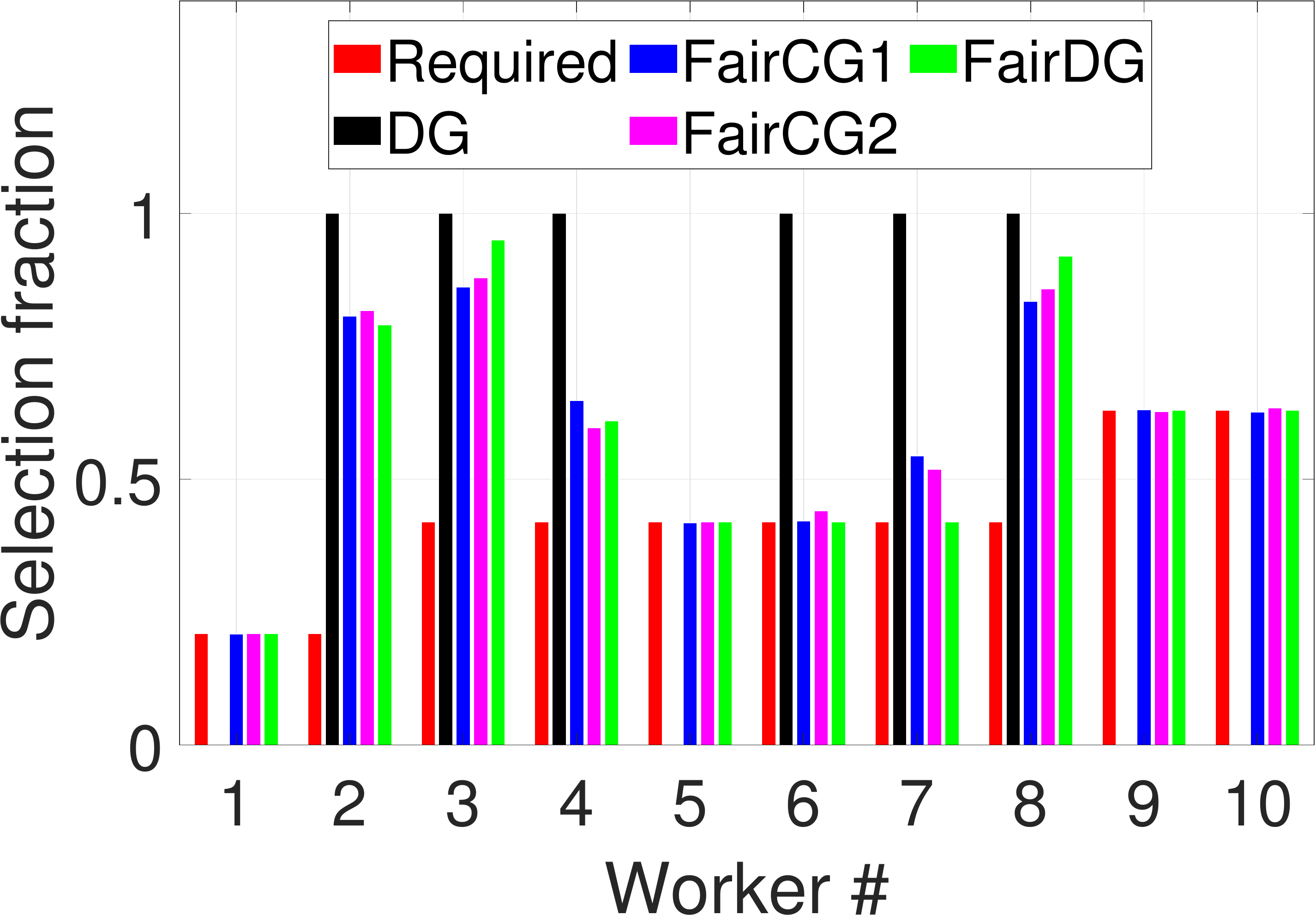}}
		\quad
		\subfigure[Selection fraction of worker $u_1$]{
			\label{fig:fraction_over_time} 
			\includegraphics[width=0.3\textwidth]{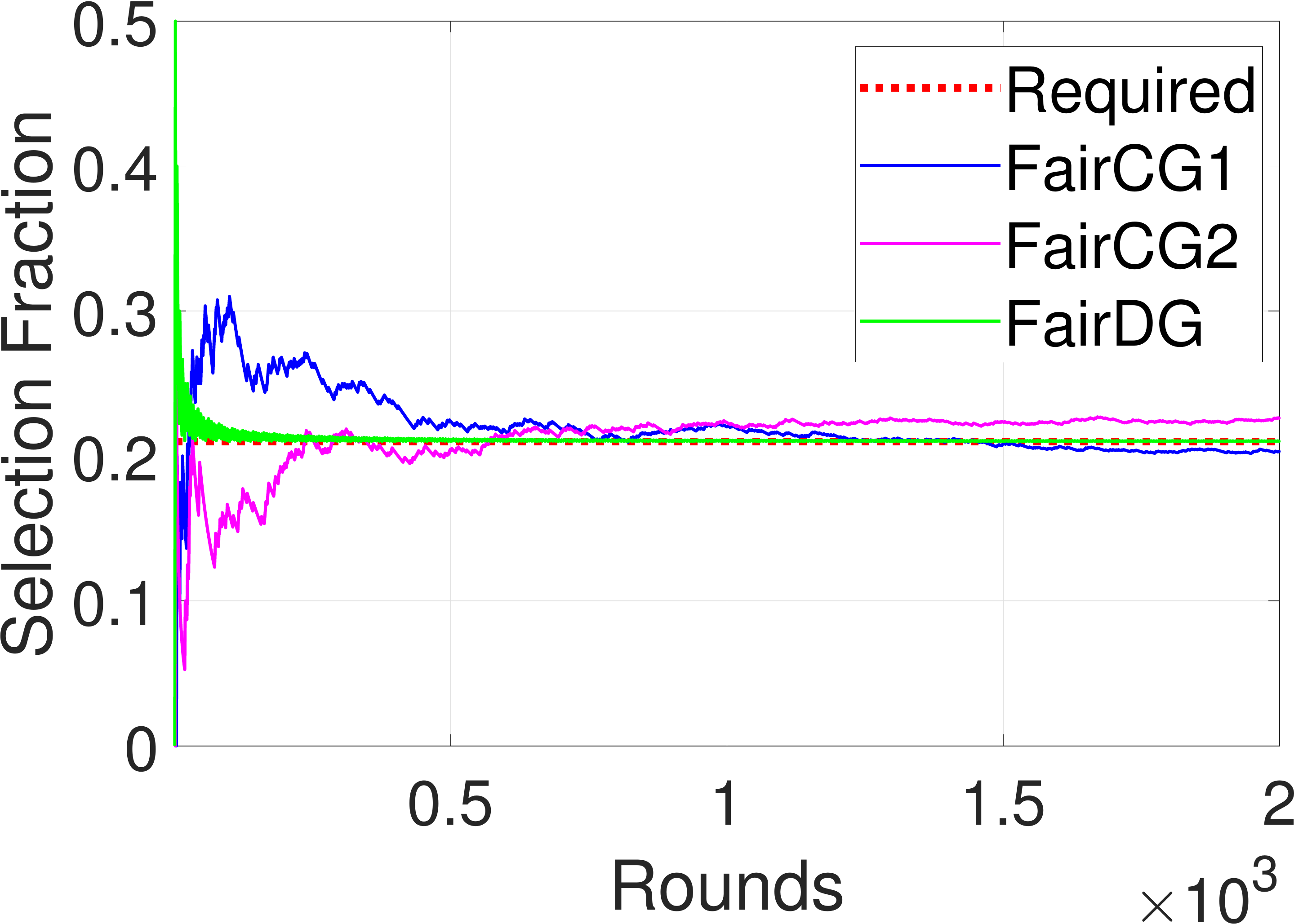}}
		\caption{Performance comparisons of different algorithms.}
		\label{fig:performance_utility_fair} 
\end{figure*}
\quad
\begin{figure}[ht]
	\begin{center}
		\centering \includegraphics[width=0.7\linewidth]{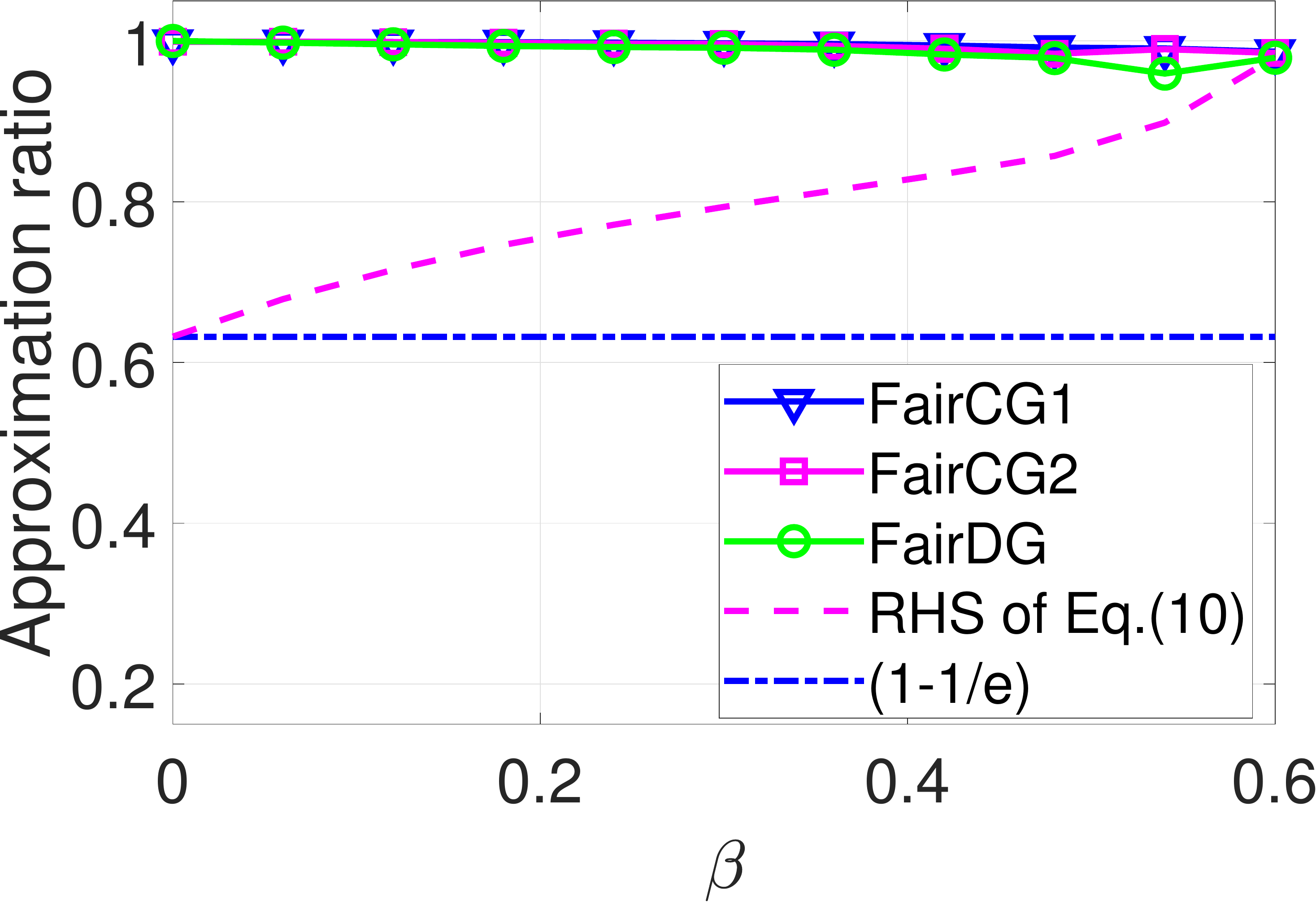}
	\end{center}
	\caption{Impact of fairness requirement on time-average utility.}\label{fig:impact_of_r}
\end{figure}

Let $N_{u,t}\triangleq\sum_{t^{\prime}=1}^t\mathbb{I}_{u\in S_{t^{\prime}}}$ denote the number of times element $u$ has been selected by the end of round $t$, 
and let $\Delta(u|S)\triangleq f(S\cup\{u\})-f(S)$ be the marginal gain of adding element $u$ to current set $S$. At the beginning of each round $t$, 
we find a set $A_t$ consisting of elements $u$ with a nonnegative ``debt'' (i.e., violation of fairness requirement), i.e., $ r_ut-N_{u,t-1}\geq 0$. Then, depending on the size of $A_t$, FairDG performs in two different ways. Let $|A_t|=l$. If there are less than $k$ violated elements, i.e., $l< k$, then we pick all of these violated elements and select $k-l$ elements from the remaining (satisfied) elements according to their marginal gain 
$\Delta(u|S)$ as in the classic discrete greedy method; if there are at least $k$ violated elements, i.e., $l\geq k$, then we select $k$ of them according to their ``debts'' in a greedy manner. 
By aggressively selecting the unsatisfied elements in each round, 
FairCG2 guarantees the short-term fairness in  Eq.~\eqref{eq:alpha_fair_requirement} 
with a homogeneous fairness requirement:
\begin{theorem} \label{thm:dg_fair}
	Assume $r_u=r$ for every element $u \in \N$, where $r\geq 0$ and $nr\leq k$. 
	The FairDG algorithm is $1$-fair.
\end{theorem}
The proof is inspired by \cite{patil2019stochastic}. We show that the fairness ``debt'' for each element $u$ by the end of round $t$ is less than one, i.e., $rt-N_{u,t}<1$ for all $u$ in $\N$ and each $t$. However, our proof is more involved because of the combinatorial nature in each round. 
Details can be found in \High{Appendix~\ref{app:short_fair_proof}.}

\textbf{Remark:} Theorem~\ref{thm:dg_fair} implies that ``debt'' of each element will go to zero as $T$ goes to infinity. This further implies that the long-term fairness requirement in Eq.~\eqref{eq:fraction requirement} can be satisfied.
FairDG satisfies a short-term fairness requirement by giving a higher priority to fairness in each round. However, this priority in fairness makes it quite challenging to analyze the time-average utility since FairDG has to select part of a set without accounting for the utility. We leave the utility analysis under FairDG for future work. Note that the optimal time-average utility $U_{\mathrm{opt}}$ for MMSM-CF with the long-term fairness requirement offers a natural upper bound on the optimal time-average utility with a short-term fairness constraint. 

FairDG has a complexity of $O^{\dagger}(knT)$, as there are $T$ rounds, and in each round, it performs at most $O^{\dagger}(kn)$ oracle queries.

\section{Numerical Results} \label{sec:simulation}
In this section, we conduct simulations to evaluate the performance of our proposed algorithms (FairCG1, FairCG2, and FairDG). Specifically, our simulations are designed to answer the following questions: i) Whether the proposed algorithms satisfy the fairness requirements?
ii) How well do the proposed algorithms perform in term of the time-average utility?
iii) How does the fairness requirement impact the achieved utility 
and the derived lower bounds of our proposed algorithms? 
iv) How tight are the theoretical lower bounds
?

\High{In simulations, we assume \emph{i.i.d.} datasets and the same computing capabilities across workers for simplicity and  use the accuracy function in \cite[Eq.~(1)]{figueroa2012predicting} by setting minimum achievable error $a=0.05$, learning rate $b=0.5$, and decay rate $c=-0.2$. Let $L_u$ be the number of samples involved in worker $u$. Then, the expected utility of selecting set $S$ is
\begin{equation}
	f(S) = (1-a) - b*\left(\sum_{u \in S} L_u\right)^c. \label{eq:dixit_func}
\end{equation}}%
Throughout the simulations, we set $n=10, k=6$, and $T=10^5$. Other parameters are presented in Table~\ref{tab:sim_setting}. The optimal average utility $U_{\mathrm{opt}}$ is obtained by running an LP solver in Matlab (\texttt{linprog}) to solve Problem~\eqref{eq:optimal_LP}.
\begin{table}[t]
	\caption{Parameters settings}
	\begin{center}
		\scalebox{0.87}{
		\begin{tabular}{c|c|c|c|c|c|c|c|c|c|c}
\hline
Worker index&$u_1$ & $u_2$ & $u_3$ & $u_4$ & $u_5$ & $u_6$ & $u_7$ & $u_8$ & $u_9$ & $u_{10}$  \\
\hline
$L_u/1000$ & $0.2$ & $ 0.8$ & $ 1$ & $ 0.5$ & $ 0.1$ & $ 0.3$ & $ 0.4$ & $ 0.9$ & $ 0.1$ & $ 0.2$  \\
\hline
$\mathbf{r}_{\mathrm{base}}$ & $0.5$ & $ 0.5$ & $ 1$ & $ 1$ & $ 1$ & $ 1$ & $ 1$ & $ 1$ & $ 1.5$ & $ 1.5$  \\
\hline
$\r$ & \multicolumn{10}{c}{$\r = \beta \r_{\mathrm{base}}$} \\
\hline
\end{tabular}
		}
	\vspace{-15pt}
	\label{tab:sim_setting}
	\end{center}
\end{table}

First, we evaluate the performance of the three proposed algorithms (FairCG1, FairCG2, and FairDG) in terms of the time-average utility and the fairness requirement satisfaction.
Let $\beta = 0.42$. Then, $\r = \beta \r_{\mathrm{base}}=[.21, .21, .42, .42, .42, .42, .42, .42, .63, .63]$, which is feasible since $\r^\mathsf{T} \mathbf{1} = 4.2 < k = 6$. For comparisons, we also consider the discrete greedy algorithm (denoted by DG). 
The simulation results are presented in Fig.~\ref{fig:performance_utility_fair}. Specifically, in Fig.~\ref{fig:utility_evolve}, we display the time-average utility over $T$ rounds for the considered algorithms, including the optimal value $U_{\mathrm{opt}}$. 
Fig.~\ref{fig:fraction} shows the selection fraction of each element at the end of $T$ rounds. 
Based on the simulation results, we make the following observations:
i) From Fig.~\ref{fig:utility_evolve}, we observe that the achieved utility of our proposed algorithms is very close to the optimal value $U_{\mathrm{opt}}$ (within 1\%). Also, Fig.~\ref{fig:fraction} indicates that the selection fraction of each worker under each of our proposed algorithms satisfies the required selection fraction. 
ii) From Fig.~\ref{fig:utility_evolve}, DG appears to achieve the largest time-average utility, which is even higher than the optimal value $U_{\mathrm{opt}}$. However, as shown in Fig.~\ref{fig:fraction}, only $k$ workers \High{$\{u_2, u_3, u_4, u_6, u_7, u_8\}$} are repeatedly selected in every round under DG, which implies that the fairness requirement is not satisfied for the other workers.

In addition, we consider the selection fraction of each worker over rounds and only present the results for worker $u_1$ over the first $2,000$ rounds in Fig.~\ref{fig:fraction_over_time} as a representative example.  
While FairCG1 and FairCG2 perform slighly better than FairDG in terms of the time-average utility as shown in Fig.~\ref{fig:utility_evolve}, we observe from Fig.~\ref{fig:fraction_over_time} that FairDG converges to a point satisfying the fairness requirement (i.e., $0.21$) much faster than FairCG1 and FairCG2 do.
This is not surprising because FairDG gives a higher priority to satisfying the fairness requirement.

Finally, we investigate the impact of the fairness requirement $\r$ on the time-average utility and the tightness of the theoretical lower bounds derived in Theorems~\ref{thm:approx_ratio_faircg1} and \ref{thm:approx_ratio_faircg2}. We set different values of the fairness requirement $\r = \beta \r_{\mathrm{base}}$ by scaling the value of $\beta$. A larger value of $\beta$ means a stronger fairness requirement. Consider $\beta \in \{0, .06, .12, .18, .24, .30, .36, .42, .48, .54, .60 \}$, i.e., we have $11$ distinct fairness requirement vectors, all of which are feasible. 
We run the proposed algorithms for each of them and plot the corresponding approximation ratio (the time-average utility over the optimal value $U_{\mathrm{opt}}$) in Fig.~\ref{fig:impact_of_r}. The ratios of the lower bounds in Theorems~\ref{thm:approx_ratio_faircg1} and \ref{thm:approx_ratio_faircg2} over $U_{\mathrm{opt}}$ are presented as well.
From Fig.~\ref{fig:impact_of_r}, we can observe that the approximation ratios under FairCG1, FairCG2, and FairDG are close to one for each fairness requirement. 
%
Interestingly, while FairCG1 is guaranteed to achieve an approximation ratio of $(1-1/e)$ uniformly for different fairness requirements (Theorem~\ref{thm:approx_ratio_faircg1}), the ratio for the theoretical lower bound of FairCG2 increases with the fairness requirement and approaches one. 
Hence, we have a tighter bound for FairCG2 as the fairness requirement becomes stronger. 
By having a lower bound with respect to fairness requirement $\r$ in Theorem~\ref{thm:approx_ratio_faircg2}, we obtain a tighter characterization for FairCG2 compared to FairCG1. 

\section{Conclusion} \label{sec:conclusion}
In this paper, we formulated the fair worker selection in FL systems as a novel problem of multi-round monotone submodular maximization with cardinality and fairness constraints. To address this new problem, we proposed three carefully designed algorithms (i.e., FairCG1, FairCG2, and FairDG)
. We presented both theoretical and simulation results to demonstrate the effectiveness of our proposed algorithms. 
Our study in this paper raises several interesting questions that are worth investigating as future work. For example, 
can we establish approximation guarantees for FairDG that satisfies the short-term fairness criterion? 
While we assume the same utility function over rounds in our model, it would be interesting to consider the setting with round-dependent (i.e., task-dependent) utility functions, which better suits certain applications. 
In this setting, we may still adopt FairCG1 and FairCG2 by applying them to each round with a different utility function, but that would incur a much higher complexity. In contrast, FairDG can be directly applied to this setting with the same complexity and the short-term fairness guarantee. However, it remains open to analyze the achieved utility under FairDG.
Finally, if the submodular function under consideration is unknown in advance, it is highly interesting to investigate the joint learning and selection problem of multi-around submodular optimization.



\bibliographystyle{IEEEtran}

\bibliography{reference}

\clearpage

\appendix
\subsection{Applications}\label{sec:app}

In this section, we discuss three additional examples of real-world applications in detail to better motivate the proposed MMSM-CF problem. These examples include sensor scheduling in wireless sensor networks \cite{gupta2006stochastic,shamaiah2010greedy,jawaid2015submodularity,tzoumas2018resilient}, task assignment in crowdsourcing platforms \cite{to2015server,yang2017identifying,gao2015providing}, and data subset selection in machine learning \cite{wei2015submodularity,kirchhoff2014submodularity,badanidiyuru2014streaming}.

\emph{Sensor Scheduling in Wireless Sensor Networks:}
We consider sensor scheduling in wireless sensor networks, where a set of inexpensive sensors are deployed to sense the environment state process. After the measurements from the sensors are transmitted to a sink node, they will be fused to estimate the environment state process. Obviously, more measurements (collected from distinct sensors) will result in a more accurate estimation of the environment state. However, only a subset of sensors could transmit their measurements simultaneously (e.g., due to wireless interference) \cite{gupta2006stochastic,shamaiah2010greedy}. Let $f(S)$ be the aggregate sensing quality of the measurements collected from the scheduled sensors $S$. 
Due to the spatial correlation among the sensor measurements, the aggregate sensing quality $f(S)$ often exhibits a diminishing returns property \cite{gupta2006stochastic,shamaiah2010greedy}. 
Moreover, one usually needs to repeatedly collect the measurements (i.e., in multiple rounds) to continuously monitor the environment and aims to maximize the overall sensing quality over time. 
In addition, to obtain a holistic view of the environment that is being monitored, one often does not want to miss out on too much data from any single sensor. This imposes a minimum delivery ratio requirement of each sensor, which can be modeled as the fairness requirement in our model. Hence, the goal is to schedule a sequence of subsets of sensors so as to maximize the overall sensing quality over time while guaranteeing a minimum delivery ratio of each sensor.

\emph{Task Assignment in Crowdsourcing Platforms:}
Crowdsourcing offers an efficient method for task distribution and completion. Consider a spatial crowdsourcing application (e.g., traffic speed estimation) \cite{to2015server, yang2017identifying}. 
When a spatial task arrives at the crowdsourcing platform, it will be assigned to a group of workers on the platform (e.g., no more than $k$ workers due to the budget limit). A certain amount of utility is generated after this particular task is completed. The utility could represent the informative data gathered for the crowdsourced sensing task. 
Due to the similarity of responses from different workers~\cite{to2015server,gao2015providing,yang2017identifying}, such as the possible overlapping sensing data from different workers~\cite{zhang2014groping}, the utility could be described as a submodular function with regard to the set of assigned workers~\cite{yang2017identifying}. The participation of more workers usually leads to more informative data and thus a larger utility (i.e., monotone). 
As in \cite{to2015server}, we assume that the sequential tasks are of the same type and that all the workers are qualified to perform the tasks. This can be captured by the multi-round nature of our model.
The goal here is to maximize the time-average utility by determining an optimal worker assignment for multiple tasks. In addition, the platform has to take fairness towards workers into account through a minimum assignment ratio guarantee for each worker. This helps maintain a healthy and sustained platform with improved satisfaction among the workers and thus encourage more participation.

\emph{Data Subset Selection in Machine Learning:}
The data subset selection problem in machine learning has been extensively studied in the literature \cite{wei2015submodularity,kirchhoff2014submodularity,badanidiyuru2014streaming}. This is motivated by both limited computational resources and redundant information in a massive amount of data. For training, one prefers to select a subset of data sources that is informative or representative of the entire dataset as modeled by the corresponding objective function.
It has been shown that some highly relevant objective functions (used to measure the informativeness or the representativeness, or the combination of the two) are submodular with regard to the selected data sources because of a diminishing returns property it exhibits.  
Consider a multi-round training process. Let the total utility be a simple additive sum of these objective values corresponding to the sequentially selected data subsets. The goal here is to maximize the total utility. 
Moreover, in order to ensure enough data for post-training data analysis, a minimum selection fraction requirement for each data source must be taken into consideration.
We should select not only the most informative or representative data sources but also those less informative or representative ones for a certain amount of times. 
This can be naturally modeled as a fairness requirement using the MMSM-CF framework. In this case, our goal is to sequentially select a subset of data sources that maximize the total utility while guaranteeing a minimum selection fraction for each data source. 

\subsection{Proof of Theorem \ref{thm:fairness_guarantee} and Theorem~\ref{thm:fairness_guarantee2}}\label{app:proof_fairness}
\begin{proof}
	First, for any feasible $\mathbf{r}$, $P_f$ is non-empty. By the end of Step 1 (i.e., $\tau=1$) of the fair continuous greedy algorithms, we have 
	
	i) FairCG1
	\begin{equation}
	  \begin{aligned}
			\y(1) &= \int_{0}^1 \frac{d\y(\tau)}{d\tau} d\tau + \y(0) \\
			&= \int_{0}^1 \x(\tau) d\tau + \mathbf{0} = \int_{0}^1 \x(\tau)d\tau, \label{eq:y_combination_x_1}
		\end{aligned}  
	\end{equation}
	
	ii) FairCG2
	\begin{equation}
		\begin{aligned}
			\y(1) &= \int_{0}^1 \frac{d\y(\tau)}{d\tau} d\tau + \y(0) \\
			&= \int_{0}^1 (\x(\tau)-\r) d\tau + \r = \int_{0}^1 \x(\tau)d\tau, \label{eq:y_combination_x_2}
		\end{aligned}
	\end{equation}
	i.e., $\y(1)$ is a convex linear combination of $\x(\tau)$'s, and thus, $\y(1)$ must be in $P_f$ since $\x(\tau) \in P_f$ and $P_f$ is a convex set. 
	Besides, it is not difficult to show $\y(1)^\mathsf{T}  \mathbf{1} = k$ since every $\x(\tau)$ is a vertex of the convex body $P_f$ satisfying  $\x^\mathsf{T}  \mathbf{1} = k$.
	Therefore, by the end of this step, we derive a fractional solution $\y(1)$ that satisfies $\r\leq \y(1)\leq \mathbf{1} ~\text{and}~ \y(1)^\mathsf{T}\mathbf{1}= k$. 
	
	Then, in each round $t$, the selected set $S_t$ is derived by performing randomized dependent rounding, i.e., $\textsc{DepRounding}$, on $\y(1)$. Note that $|S_t|=k$ because of the loop invariant $\y^{\mathsf{T}}\mathbf{1}=k$ in $\textsc{DepRounding}$.
	Let a random variable $Y_u(t)\in \{0,1\}$ represent whether element $u$ is selected or not in round $t$, i.e., $Y_u(t)=\mathbb{I}_{\{u \in S_t\}}$. We have $\mathbb{P}\{Y_u(t)=1\}=y_u(1)$ and $\sum_{u\in \N} Y_u(t) = \sum_{u\in \N} y_u(1) = k$ according to the \emph{marginal distribution} property and the \emph{degree-preservation} property of dependent rounding scheme respectively in \cite{gandhi2006dependent}. Eventually, for the stationary randomized algorithm, we have $\liminf_{T\to \infty} \frac{1}{T}\sum_{t=1}^{T}\E\left[\mathbb{I}_{\{u\in S_t\}}\right] = \E\left[ \mathbb{I}_{\{u \in S_t\}}\right]=\mathbb{P}\{Y_u(t)=1\} = y_u(1)\geq r_u$ for every element $u$. The fairness requirement in Eq.~\eqref{eq:fraction requirement} is satisfied. 
	
	Consider $T>0$ rounds. Besides, for every element $u$, $Y_u(1), Y_u(2), \dots, Y_u(T)$ are \emph{i.i.d.} Bernuoulli variables with mean $y_u(1)$. According to the Hoeffding Inequality \cite{hoeffding1994probability}, the empirical mean of these variables 
     	satisfies, for any $\delta>0$,
	\begin{equation}
		\mathbb{P}\left\{\frac{1}{T}\sum_{t=1}^{T} Y_u(t) - y_u(1)  \leq - \delta\right\} \leq e^{-2T\delta^2},
	\end{equation}
	which further indicates the inequality in Eq.~\eqref{eq:results_fairness} since $Y_u(t) =\mathbb{I}_{u\in S_t}$ and $y_u(1)\geq r_u$ for each element $u$. 
\end{proof}

\subsection{Proof of Theorem~\ref{thm:approx_ratio_faircg1}}
Before proving Theorem~\ref{thm:approx_ratio_faircg1}, we first show an important inequality in the following Lemma \ref{lem:utility_inequality}.

\begin{lemma} \label{lem:utility_inequality}
	Consider any time point $\tau$ in Step 1 of FairCG. Combining the definitions of $\mathbf{w}(\tau)$ and $\x(\tau)$, 
	we have the following inequality,
	\begin{equation}
		\x(\tau)^\mathsf{T} \mathbf{w}(\tau) \geq U_{\mathrm{opt}} -F(\y(\tau)). \notag
		\label{eq:lemma_for_approx}
	\end{equation} 
\end{lemma}
The proof for Lemma \ref{lem:utility_inequality} is shown in Appendix \ref{app:proof_lemma_ineq}.

\begin{proof}[Proof of Theorem~\ref{thm:approx_ratio_faircg1}]
	To derive the result in Theorem \ref{thm:approx_ratio_faircg1}, we first show that the fractional vector $\y(1)$ satisfies $F(\y(1)) \geq (1-1/e) U_{\mathrm{opt}}$ (\textcircled{1}) and then prove $\mathbb{E}[f(S_t)]\geq F(\y(1))$ in each round $t\in \{1,2, \cdots\}$ (\textcircled{2}). Combining the conditions \textcircled{1} and \textcircled{2}, we derive the result in Eq.~\eqref{eq:approx_ratio_faircg1}. 
	
	First, we show \textcircled{1}: $F(\y(1)) \geq (1-1/e) U_{\mathrm{opt}}$. 
	
	The error due to discretization could be made (polynomially) small 
	\cite{vondrak2008optimal}, and thus, we here only give an analysis for the continuous version. Starting with $\y(0)=\mathbf{0}$ and $F(\y(0))=F(\mathbf{0})=0$, we want to see how much $F(\cdot)$ increases during each discretized time interval $[\tau,\tau+d\tau)$ in the \texttt{while} loop (Lines 
	2-9). Applying the chain rule yields:
	\begin{equation}
		\begin{aligned}
			&\frac{dF(\y(\tau))}{d\tau} = \sum_{u \in \N}\left(\frac{dy_u(\tau)}{d\tau}\cdot \frac{\partial F(\y)}{\partial y_u}\bigg|_{\y = \y(\tau)}\right)\\
			=&\sum_{u \in \N}\left(x_u(\tau)\cdot \frac{\partial F(\y)}{\partial y_u}\bigg|_{\y = \y(\tau)}\right) \\
			=&\sum_{u \in \N}\left(x_u(\tau)\cdot \frac{F(\y(\tau)\vee \mathbf{1}_u)-F(\y(\tau))}{1-y_u(\tau)}\right)\\
			\geq& \sum_{u \in \N} x_u(\tau) \cdot w_u(\tau)\\
			=& \x(\tau)^\mathsf{T} \mathbf{w}(\tau). \label{eq:derivative_F1}
		\end{aligned}
	\end{equation}

	Combining the result in Lemma \ref{lem:utility_inequality}, we have the differential inequality with respect to the function of $F(\y(\tau))$: 
	$$
	\frac{dF(\y(\tau))}{d\tau} \geq U_{\mathrm{opt}}- F(\y(\tau)).
	$$
	Solving the above differential inequality with the initial condition $F(\y(0))=0$ under FairCG1, we have the fractional vector $\y(1)$ satisfying
	\begin{equation}
		F(\y(1)) \geq (1-1/e) U_{\mathrm{opt}}. \label{eq:Fy_bound_faircg1}
	\end{equation}
	
	
	Then, we show
	$\mathbb{E}[f(S_t)]\geq F(y(1))$ by employing the property of dependent rounding and convexity of the multilinear extension $F(\cdot)$ in any direction $\mathbf{d} = \mathbf{1}_u-\mathbf{1}_v$ for any pair of distinct elements $u,v\in \N$. Here, the expectation is taken over the randomness of selecting set with FairCG.
	
	The randomized dependent rounding process proceeds as follows. It starts with an $n$-dimensional fractional vector $\y$ and rounds at least one floating element $y_u\in (0,1)$ in each iteration of the \texttt{While} loop. 
	Hence, the \texttt{while} loop takes at most $n$ iterations. 
	
	Let $\z_{m}$ be the random variable denoting the value of $\y$ at the beginning of iteration $m$, and let $n_0$ be the last iteration after which all the elements of $\y$ are integers ($n_0\leq n$). We have $\z_1 = \y(1)$ and the returned set $S= \{u\in \N: (Z_{n_0+1})_u=1)\}$. 
	In the following, we will first show that 
	\begin{equation}
		\forall m, \E[F(\z_{m+1})] \geq \E[F(\z_{m})],  \label{eq:depR_updating_guarantee}
	\end{equation}
	where the expectation is taking over the randomness 
	of the updating step in Line 16 in Algorithm~\ref{alg:faircg1}.
	Then, we have $$
	\E[f(S)] =\E[F(\z_{n_0+1})]\geq \E[F(\z_1)]= F(\y(1)).
	$$
	
	We now prove the inequality Eq.~\eqref{eq:depR_updating_guarantee} for a fixed iteration $m$. Suppose that elements $u$ and $v$ with $y_u, y_v \in (0,1)$ are the two elements found in current iteration (Line~14 of Algorithm~\ref{alg:faircg1}. Let $a_m$ and $b_m$ be the values of $a$ and $b$ respectively in iteration $m$, and $\mathbf{d}_m = \mathbf{1}_u-\mathbf{1}_v$. Then, we have $\z_{m+1}$ equal to $\z_m+a_m \mathbf{d}_m$ with probability $\frac{b_m}{a_m+b_m}$ and equal to $\z_m-b_m \mathbf{d}_m$ with probability $\frac{a_m}{a_m+b_m}$, which indicates the conditional expectation given $F(\z_m)=h$, 
	\begin{equation}
		\begin{aligned}
			&\E[F(\z_{m+1})|F(\z_m)=h] \\
			=& \frac{b_m}{a_m+b_m} F(\z_m+a_m \mathbf{d}_{m})+ \frac{a_m}{a_m+b_m}F(\z_m-b_m \mathbf{d}_m).
		\end{aligned}
	\end{equation}
	Let $g_\z(\xi)\triangleq F(\z_m+ \xi \mathbf{d}_m)$. Now $g_\z(\xi)$ is convex in $\xi$ due to the convexity of $F(\cdot)$ in the direction $\mathbf{d}_m = \mathbf{1}_u-\mathbf{1}_v$.
	It indicates: $ \frac{b_m}{a_m+b_m}g_\z(a_m) + \frac{a_m}{a_m+b_m}g_\z(-b_m)\geq g_\z(0)$, i.e., $\frac{b_m}{a_m+b_m} F(\z_m+a_m \mathbf{d}_{m})+ \frac{a_m}{a_m+b_m}F(\z_m-b_m \mathbf{d}_m)\geq F(\z_m)$.
	Hence, we have $$\E[F(\z_{m+1})|F(\z_m)=h]\geq F(\z_m)=h.$$
	Let $H$ be the set of all possible values of $h$. We have the following inequalities:
	\begin{equation}
		\begin{aligned}
			&\E[F(\z_{m+1})] \\
			= &\sum_{h\in H} \E[F(\z_{m+1})|F(\z_m)=h]\cdot \mathbb{P}\{F(\z_m)=h\} \\
			\geq & \sum_{h\in H}h\cdot \mathbb{P}\{F(\z_m)=h\} \\
			=& \E[F(\z_m)],
		\end{aligned}
	\end{equation}
	which is exactly Eq.~\eqref{eq:depR_updating_guarantee}. This completes our proof for $\E[f(S)]\geq F(\y(1))$ and then the results in Eq.~\eqref{eq:approx_ratio_faircg2} by combining the result in Eq.~\eqref{eq:Fy_bound}. 
\end{proof}

\subsection{Proof of Lemma \ref{lem:utility_inequality}} \label{app:proof_lemma_ineq}
\begin{proof}
	Let $\mathbf{1}_S$ represent the characteristic vector of subset $S$ and $\q^*=[q^*_S]_{S\in \N_k}$ denote an optimal solution to Problem~\eqref{eq:optimal_LP}. Define an $n$-dimensional vector $\y^{*}\triangleq \sum_{S\in \N_k}q_S^* \mathbf{1}_S$ with the coordinate corresponding to $u$ being $y^*_u = \sum_{S\in \N_k: u \in S} q_S^{*}$ for each $u$ in $\N$. 
	Then, we have $\r\leq \y^*\leq \mathbf{1}$ (constrains in Problem~\eqref{eq:optimal_LP}) and ${\y^*}^\mathsf{T}\mathbf{1}\leq k$ since any set $S\in \N_k$ satisfies $\mathbf{1}_S^\mathsf{T}\mathbf{1}\leq k$, which implies $\y^*\in P_f$. 
	Based on the definition of $\x(\tau)$, we have
	\begin{equation}
		\begin{aligned}
			&\x(\tau)^\mathsf{T}  \mathbf{w}(\tau)	\geq  {\y^*}^\mathsf{T} \mathbf{w}(\tau) \\
			= &\sum_{u \in \N} y^*_u \left( F(\y(\tau)\vee \mathbf{1}_u)-F(\y(\tau))\right)  \\
			=&\sum_{u \in \N} \sum_{S\in\N_k:u \in S} q^*_S  \left( F(\y(\tau)\vee \mathbf{1}_u)-F(\y(\tau))\right) \\
			= & \sum_{S\in\N_k} q^*_S  \sum_{u \in S} \left( F(\y(\tau)\vee \mathbf{1}_u)-F(\y(\tau))\right).
		\end{aligned} \notag
	\end{equation}
	Recall that the multilinear extension $F(\y)$ represents the expected value of the submodular function $f(R(\y))$ where $R(\y)$ is a random set with each element $u$ being independently selected with probability $y_u$. To distinguish from other randomness, we denote the multilinear extension as $F(\y) = \E_{R\thicksim \y}[f(R(\y))]$. We have
	\begin{equation}
		\begin{aligned}
			& \sum_{S\in\N_k} q^*_S  \sum_{u \in S} \left( F(\y(\tau)\vee \mathbf{1}_u)-F(\y(\tau))\right) \\
			=& \sum_{S\in \N_k} q^*_S  \sum_{u \in S} \mathbb{E}_{R\thicksim \y}[f(R(\y(\tau)\vee \mathbf{1}_u))-f(R(\y(\tau)))]\\
			\overset{(a)}{\geq}& \sum_{S\in\N_k} q^*_S    \mathbb{E}_{R\thicksim \y}[f(R(\y(\tau)\vee \mathbf{1}_{S}))-f(R(\y(\tau)))]  \quad \\  
			\overset{(b)}{\geq}& \sum_{S\in\N_k} q^*_S  \mathbb{E}_{R\thicksim \y}[f(R(\mathbf{1}_{S}))-f(R(\y(\tau)))] \\
			=& \sum_{S\in\N_k} q^*_S  (f(S)-F(\y(\tau)))  \quad \\ 
			=& \sum_{S\in\N_k} q^*_S  f(S)- \sum_{S\in \N_k} q^*_SF(\y(\tau))) \\
			=& U_{\mathrm{opt}}- F(\y(\tau))), \notag
		\end{aligned}
	\end{equation}
	where 
	inequality $(a)$ holds due to the submodularity of $f$ and $(b)$ due to its monotonicity. 
	
	Then, we derive the result Eq.~\eqref{eq:lemma_for_approx} in Lemma~\ref{lem:utility_inequality}. 
\end{proof}

\subsection{Proof of Theorem \ref{thm:approx_ratio_faircg2}}\label{app:proof_approx}
\begin{proof}
	Similar to the proof of Theorem~\ref{thm:approx_ratio_faircg1}, we first show that the fractional vector $\y(1)$ satisfies $F(\y(1)) \geq (1-1/e^{c_{\mathbf{r}}}) U_{\mathrm{opt}} +F(\mathbf{r})/e^{c_{\mathbf{r}}}$ (\textcircled{1}) and then that $\mathbb{E}[f(S_t)]\geq F(\y(1))$ holds in each round $t=\{1, 2, \cdots\}$, (\textcircled{2}), which is exactly the same as Theorem~\ref{thm:approx_ratio_faircg1}. Combining them, we derive the result in Eq.~\eqref{eq:approx_ratio_faircg2}. 
	
	In this proof, we will show
	$F(\y(1))\geq (1-1/e^{c_{\mathbf{r}}}) U_{\mathrm{opt}} +F(\mathbf{r})/e^{c_{\mathbf{r}}}$ under FairCG2. 
	
FairCG2 starts with $\y(0)=\r$ and updates $\y(\tau)$ with a rate $\x(\tau)-\r$. As the previous proof, we will see how much $F(\cdot)$ increases during each discretized time interval $[\tau,\tau+d\tau)$ in the \texttt{while} loop (Lines 
	2-9). By aplying the chain rule, we have
	\begin{equation}
		\begin{aligned}
			&\frac{dF(\y(\tau))}{d\tau} = \sum_{u \in \N}\left(\frac{dy_u(\tau)}{d\tau}\cdot \frac{\partial F(\y)}{\partial y_u}\bigg|_{\y = \y(\tau)}\right)\\
			=&\sum_{u \in \N}\left(\left(x_u(\tau)-r_u\right)\cdot \frac{\partial F(\y)}{\partial y_u}\bigg|_{\y = \y(\tau)}\right) \\
			=&\sum_{u \in \N}\left(\left(x_u(\tau)-r_u\right)\cdot \frac{F(\y(\tau)\vee \mathbf{1}_u)-F(\y(\tau))}{1-y_u(\tau)}\right)\\
			\geq& \sum_{u \in \N} \left(x_u(\tau)-r_u\right) \cdot w_u(\tau) = (\x(\tau)-\r)^\mathsf{T} \mathbf{w}(\tau). \label{eq:derivative_F2}
		\end{aligned}
	\end{equation}
	Let $\y^{\prime}(\tau)$ be a point in the convex polytope $P_f$, such that $\y^{\prime}(\tau)-\r = c_{\tau} \x(\tau)$, where $c_{\tau}$ is a constant. Then, $\y^{\prime}(\tau)$ satisfies the following two constraints:
	$$
	\begin{cases}
	y^{\prime}_u(\tau)\leq 1, \forall u \\
	\sum_{u\in \N} y^{\prime}_u(\tau) \leq k,
	\end{cases}
	$$
	i.e.,
	\begin{equation}
		\begin{cases}
			r_u +c_{\tau}x_u(\tau)\leq 1, \forall u \\
			\sum_{u\in \N} r_u +c_{\tau}k \leq k, \label{eq:ct_constraints}
		\end{cases}
	\end{equation}
	where the second inequality of Eq.~\eqref{eq:ct_constraints} is from $\sum_{u \in \N} x_u(\tau)=k$ due to the fact that $\x(\tau)$ must be a vertex of the convex polytope $P_f$. Let $D_{\tau} = \{u\in \N: x_u(\tau)\neq 0\}$. Combining the constraints in Eq.~\eqref{eq:ct_constraints}, we set $c_{\tau} = \min\{\min\limits_{u\in D_{\tau}}\frac{1-r_u}{x_u(\tau)}, 1-\frac{\sum_{u\in \N}r_u}{k}\}$
	and let $c_{\mathbf{r}}=1-\max\{\max\limits_u r_u, \frac{\sum_{u} r_u}{k}\}$. Obviously, we have $c_{\tau}\geq c_{\mathbf{r}}$ and thus $\y^{\prime}(\tau)-\r\geq c_{\r}\x(\tau)$. Then, the inequality in Eq.~\eqref{eq:derivative_F2} becomes
	\begin{equation}
		\begin{aligned}
			&\frac{dF(\y(\tau))}{d\tau} \\
			\geq& (\x(\tau)-\r)^\mathsf{T} \mathbf{w}(\tau) \\
			\geq&(\y^{\prime}(\tau)-\r)^\mathsf{T}  \mathbf{w}(\tau) \\
			\geq& c_{\mathbf{r}} \x(\tau)^\mathsf{T}  \mathbf{w}(\tau). \notag
		\end{aligned}
	\end{equation}
	where the second inequality is from the definition of $\x(\tau)$. 

	Combining Lemma \ref{lem:utility_inequality}, we have the differential inequality with respect to the function of $F(\y(\tau))$: 
	$$
	\frac{dF(\y(\tau))}{d\tau} \geq c_{\r} (U_{\mathrm{opt}}- F(\y(\tau))).
	$$
	Solving the above differential inequality with the initial condition $F(\y(0))=F(\r)$ under FairCG2, we have the fractional vector $\y(1)$ satisfying
	\begin{equation}
		F(\y(1)) \geq (1-1/e^{c_{\mathbf{r}}}) U_{\mathrm{opt}} +F(\mathbf{r})/e^{c_{\mathbf{r}}}. \label{eq:Fy_bound}
	\end{equation}

Combining the result $\mathbb{E}[f(S_t)]\geq F(\y(1))$ in each round $t$, we derive the result in Eq.~\eqref{eq:approx_ratio_faircg2}.
\end{proof}

\subsection{Proof of Theorem \ref{thm:dg_fair}}\label{app:short_fair_proof}

\High{
Before proving Theorem~\ref{thm:dg_fair}, we first present a sufficient and necessary condition for a fairness requirement vector $\r$ being feasible. The proof is shown in Appendix~\ref{app:proof_feasible_r}.
\begin{lemma}\label{lem:feasible_r}
	A fairness requirement vector $\r$ is feasible if and only if $\r^\mathsf{T}\mathbf{1} \leq k$.
\end{lemma}
}
The proof of Theorem~\ref{thm:dg_fair} is inspired by \cite{patil2019stochastic}. We show that the fairness ``debt'' for each element $u$ by the end of round $t$ is less than one, i.e., $rt-N_{u,t}<1$ for all $u$ in $\N$ and each $t$. However, our proof is more involved because of the combinatorial nature in each round.
Specifically, in \cite{patil2019stochastic}, only one element could be selected in each round, and a feasible requirement $\r$ satisfies $nr\leq 1$. They divide the interval $(-\infty, nr)$ into $n+1$ partitions by the points in $\{0,r, 2r, \cdots, (n-1)r\}$ and show that the fairness ``debt'' of each element must lie in one of the partitions in every round, which further implies $rt-N_{u,t}<nr\leq 1$. During this process, it is not difficult to identify the new partition to which each element belongs after the selection in each round. However, for the MMSM-CF problem we consider, a fairness requirement $\r$ is feasible only if $nr\leq k$ (Lemma~\ref{lem:feasible_r}). We cannot do the partitions in a similar manner since $n/k$ is not necessarily an integer. Moreover, it sometimes becomes difficult to identify the new partition to which each element belongs after the selection. In our proof, we divide the time horizon into small periods and group the elements in $\N$ according to their selection times. Then, we show by induction that two useful conditions about the groups hold in every period, which further implies the short-term fairness guarantee.
\begin{proof} 

We prove Theorem~\ref{thm:dg_fair} by showing that $\frac{1}{T}\sum_{t=1}^{T} \mathbb{I}_{\{u\in S_{t}\}} > r- \frac{1}{T}$ holds for every element $u$ and any $ T\in \{1,2,\dots\}$. That is, $rT-N_{u,T}< 1$ holds for every element $u$ and any $ T\in \{1,2,\dots\}$.
	We call the value $d_{u,T} \triangleq rT-N_{u,T}$ of element $u$ as the fairness debt of element $u$ by the end of round $T$. In the proof, we show that this debt is less than one for every element $u$ and any $T=1,2,\cdots$.
	
	Note that the fairness debt of each element $u$ in round $T$ is determined by its being selected times $N_{u,T}$ (because we assume the same fairness requirement $r$ for every $u$). We focus on finding the potential trend in terms of the selection times of all elements with FairDG. 
	In round $T$, we partition the ground set $\N$ according to the selected times of each element $N_{u,T}$. Define
	$J_{T,m} \triangleq \{u: N_{u,T}=m\}$. Then, we have $\N = J_{T,0}\cup J_{T,1}\cup \cdots \cup J_{T,T}$ for any $T\in \{1,2, \cdots\}$. Moreover, for simplicity, we define the following sets:
	\begin{align}
		\bar{J}_{T,m} \triangleq \{u: N_{u,T}> m\}, \notag\\
		\underline{J}_{T,m} \triangleq \{u: N_{u,T}< m\}. \notag
	\end{align}
	
	Then the following lemma is the key of the proof, implying the short-term fairness in Theorem~\ref{thm:dg_fair}.
	
	\begin{lemma}\label{lem:inductive_rslt}
		For any integer $m\in \{0,1,2, \cdots\}$
		, we have 
		
		$1.~J_{\lceil\frac{m}{r}\rceil,m}\cup \bar{J}_{\lceil\frac{m}{r}\rceil,m}=\N,$
		
		$2.~\vert J_{\lceil\frac{m}{r}\rceil,m} \vert \leq 
		\left(\frac{m+1}{r}- \lceil\frac{m}{r}\rceil \right) k.$
	\end{lemma}
	
	We provide the proof of Lemma~\ref{lem:inductive_rslt} in Appendix~\ref{app:proof_lemma_debt}.
	
	Condition $1$ in Lemma~\ref{lem:inductive_rslt} ensures that each element is selected for at least $m$ times by the end of round $\lceil\frac{m}{r}\rceil$. 
	
 Consider an arbitrary integer $m\geq 0$. For any round $T$ such  that $\lceil\frac{m}{r}\rceil \leq T \leq  \lceil\frac{m+1}{r}\rceil-1$, and any element $u$, we have $N_{u,T}\geq m$ and the fairness debt of $u$ satisfying
	\begin{equation}
		\begin{aligned}
			d_{u,T}=& rT- N_{u,T}  \\  \leq & 
			r\left(\left\lceil\frac{m+1}{r}\right\rceil-1\right)- m \\
			< & r \left(\frac{m+1}{r}\right) - m = 1.
		\end{aligned}
	\end{equation}
	Therefore, we have $d_{u,T}<1$ for every $T\in\{1,2,\cdots\}$, implying FairDG being $1$-fair.
\end{proof}
\subsection{Proof of Lemma~\ref{lem:feasible_r}}\label{app:proof_feasible_r}
\begin{proof}
	First, we prove that the condition, $\r^\mathsf{T}  \mathbf{1}\leq k$, is necessary: if $\r$ is feasible, we have $\r^{\mathsf{T}}\mathbf{1}\leq k$. 
	
	If the requirement vector $\r$ is feasible, there exists a policy that schedules a sequence of sets $\mathcal{S} = (S_1, S_2, \dots)$ 
	satisfying the fairness requirement in Eq.~\eqref{eq:fraction requirement},
	which implies
	\begin{equation}
		\begin{aligned}
			\sum_{u\in\N}r_u\leq &\sum_{u\in\N}\liminf_{T\to \infty}\frac{1}{T} \sum_{t=1}^{T} \E\left[\mathbb{I}_{\{u \in S_t\}}\right]\\
			\leq & \liminf_{T\to \infty}\frac{1}{T} \sum_{t=1}^{T}  \E\left[\sum_{u\in\N}\mathbb{I}_{\{u \in S_t\}}\right]\\
			\overset{(a)}{\leq} & \liminf_{T\to \infty}\frac{1}{T} \sum_{t=1}^{T} k=k,
		\end{aligned}
	\end{equation}
	where (a) is from the cardinality constraint. 
	Hence, we have 
	$\sum_{u\in \N} r_u \leq k$, i.e., $\r^\mathsf{T}  \mathbf{1} \leq k$.
	
	Then, we show that the condition $\r^\mathsf{T}  \mathbf{1}\leq k$ is also sufficient. That is, we can always find a policy satisfying the fairness requirement in Eq.~\eqref{eq:fraction requirement} as long as $\r^\mathsf{T}  \mathbf{1}\leq k$. Consider $T>0$ rounds.
	As described in the model, in each round we can select $k$ elements. Thus, we can treat these $T$ rounds as $kT$ slots, where we can select one element in each slot. Denote the $kT$ slots as $(1^{(1)}, \cdots, T^{(1)}, 1^{(2)}, \cdots, T^{(2)}, \cdots, 1^{(k)}, \cdots, T^{(k)})$ and the ground set as $\N=\{u_1, u_2, \cdots, u_n\}$.
	Consider a policy $\pi^{\prime}$ that, starting from slot $1^{(1)}$, assigns each element $u_i$ for consecutive slots one by one until the $\lceil(\sum_{j=1}^i r_{u_j})T\rceil$-th slot. Specifically, assign the first element $u_1$ in the first $\lceil r_{u_1}T\rceil$ slots, the second element $u_2$ in the following $\lceil (r_{u_2}+r_{u_1})T\rceil -\lceil r_{u_1}T\rceil$ slots, and so on. 
	Since $\r^\mathsf{T}  \mathbf{1}\leq k$, i.e., $\lceil\sum_{u_i\in \N} r_{u_i} T\rceil \leq \lceil kT \rceil=kT$, policy $\pi^{\prime}$ completes the above assignment for the last element $u_n$, and each element $u_i$ is assigned in $\lceil(\sum_{j=1}^i r_{u_j})T\rceil-\lceil(\sum_{j=1}^{i-1} r_{u_j})T\rceil$ slots.
	Then, for any $t\leq T$, select the elements that are assigned in the slots, $t^{(1)}, t^{(2)}, \cdots, t^{(k)}$, in round $t$. For those rounds with less than $k$ elements, add any element that has not been selected in that round. Note that each element $u_i$ will be scheduled in distinct rounds 
	due to $\lceil(\sum_{j=1}^i r_{u_j})T\rceil-\lceil(\sum_{j=1}^{i-1} r_{u_j})T\rceil\leq \lceil r_{u_i}T\rceil \leq T$. Taking limits yields that the selection fraction of each element $u_i$ satisfies
	\begin{equation}
		\begin{aligned}
			&\liminf_{T\to \infty}\frac{1}{T} \sum_{t=1}^{T} \E\left[\mathbb{I}_{\{u_i \in S_t\}}\right]\\
			\geq &\liminf_{T\to \infty}\frac{\left\lceil(\sum_{j=1}^i r_{u_j})T\right\rceil-\left\lceil(\sum_{j=1}^{i-1} r_{u_j})T\right\rceil}{T} 
			\\
			\geq& \liminf_{T\to \infty}\frac{\left\lceil r_{u_i}T\right\rceil-1 }{T}= r_u,  
		\end{aligned}
	\end{equation}
	Therefore, with the condition $\r^\mathsf{T}  \mathbf{1}\leq k$, policy $\pi^{\prime}$ satisfies the fairness requirements Eq.~\eqref{eq:fraction requirement}.
\end{proof}
\subsection{Proof of Lemma~\ref{lem:inductive_rslt}}\label{app:proof_lemma_debt}

We prove Lemma~\ref{lem:inductive_rslt} by induction. In the inductive step, we assume that the two conditions hold for any $m\in \{0,1,2,\cdots\}$ and then show that they still hold for $m+1$. First, we present two important results based on the the assumption and then directly use them in the inductive step. 

We divide the time horizon into small periods at $\lceil\frac{m}{r}\rceil$, where $m=0,1,2,\cdots$. 
Denote the $m$-th period as $\mathcal{T}_m$, i.e., $\mathcal{T}_m = \{\lceil\frac{m-1}{r}\rceil+1, \cdots, \lceil\frac{m}{r}\rceil\}$ for $m=1,2, \cdots$.

Assume that the two conditions hold for $m$. We have $N_{u,\lceil\frac{m}{r}\rceil}\geq m$ for every element $u$ and the following three observations. 

\textbf{Observation 1: } For any element $u$ with $N_{u,t-1}=m$, i.e., $u\in J_{t-1,m}$, we have the following observations.

$1.~$ if $u\in S_t$ (selected), then $u\in J_{t,m+1}$,

$2.~$ if $u\notin S_t$ (not selected), then $u\in J_{t,m}$. 

\textbf{Observation 2: }
Assume $N_{u,\lceil\frac{m}{r}\rceil}\geq m$ for every $u$ in $\N$. For $t\in \mathcal{T}_{m+1}\backslash \{\lceil\frac{m+1}{r}\rceil\}$, we have the following results:

$1.~ A_t =      J_{t-1,m}$.

$2.~J_{t, m}\subseteq J_{t-1,m} $

$3. $ If $|A_t|\geq k$, then $|J_{t,m}|=|J_{t-1,m}|-k, \text{and } |J_{t,m+1}|=|J_{t-1,m+1}|+k$.

$4. $ If $|A_t|<k$, then $|J_{t,m}|=0, \text{and } |J_{t,m+1}|\leq |J_{t-1,m}|+ |J_{t-1,m+1}|$.

\begin{proof}
	For any $t=\lceil\frac{m}{r}\rceil +1,\cdots, \lceil\frac{m+1}{r}\rceil-1$, we first have $N_{u,t-1}\geq m$ for any $u$ in $\N$, i.e., $J_{t-1,m}\cup \bar{J}_{t-1,m}=\N$.
	
	$1.~$ 
	i) $\forall u\in J_{t-1,m}$, 
	
	$$rt-N_{u,t-1}>r\lceil\frac{m}{r}\rceil-m\geq 0$$
	
	$$\Rightarrow u\in A_t.$$
	
	ii) $\forall u\in \bar{J}_{t-1,m}$, $N_{u,t-1}\geq m+1$ and
	
	$$
	\begin{aligned}
	rt-N_{u,t-1}&\leq r\left(\lceil\frac{m+1}{r}\rceil-1\right)-(m+1)\\
	&< r\cdot \frac{m+1}{r}-(m+1)=0
	\end{aligned}
	$$
	
	$$\Rightarrow u\notin A_t.$$
	Then, we have the result $1.$
	
	$2.$ Notice that $J_{t-1,m}\cup \bar{J}_{t-1,m}=\N$. 
	
	i) $\forall u\in J_{t-1,m}$, 
	\begin{equation}
		u\in
		\left\{
		\begin{array}{ll}
			J_{t,m+1}, &  \text{if~} u\in S_t,\\
			J_{t,m}, &\text{if~} u\notin S_t.
		\end{array}
		\right.
	\end{equation}
	i.e.,
	\begin{equation}
		u\in
		\left\{
		\begin{array}{ll}
			\bar{J}_{t,m}, &  \text{if~} u\in S_t,\\
			J_{t,m}, &\text{if~} u\notin S_t.
		\end{array}
		\right.
	\end{equation}
	ii) $\forall u\in \bar{J}_{t-1,m}$, we have
	\begin{equation}
		u\in \bar{J}_{t,m},
	\end{equation}
	whenever $u$ is selected or not.

	From above, we know that every element $u$ in $J_{t,m}$ is from $J_{t-1,m}$ and then the result holds.
	
	$3.$
	if $\vert A_t\vert \geq k$, then $S_t\subseteq A_t$ according to FairDG, indicating that $k$ elements from $J_{t-1,m}$ are selected as $S_t$. Hence, the selection times of element $u$ is
	\begin{equation}
		\begin{aligned}
			&N_{u,t}
			=
			\left\{
			\begin{array}{ll}
				m+1, &  \forall u\in S_t,\\
				m, &\forall u \in J_{t-1,m}\backslash S_t,\\
				N_{u,t-1}\geq m+1, &\forall u\in \bar{J}_{t-1,m}.
			\end{array}
			\right.
		\end{aligned}
	\end{equation}
	We have $J_{t,m}=J_{t-1,m}\backslash S_t$ and then 
	$|J_{t,m}|=|J_{t-1,m}|-k$. Similarly, $J_{t,m+1}=J_{t-1,m+1}\cup S_t$ and then 
	$|J_{t,m+1}|=|J_{t-1,m+1}|+k$.
	
	$4.$ If $|A_t|<k$, then every element in $A_t$ is selected in round $t$, i.e., $A_t\subseteq S_t$. Hence, the selection times of element $u$ is
	\begin{equation}
		\begin{aligned}
			&N_{u,t}
			=
			\left\{
			\begin{array}{ll}
				m+1, &  \forall u\in J_{t-1,m},\\
				N_{u,t-1}+1\geq m+2, &\forall u \in S_t\backslash J_{t-1,m},\\
				N_{u,t-1}\geq m+1, &\forall u\in \bar{J}_{t-1,m}\backslash S_t.
			\end{array}
			\right.
		\end{aligned}
	\end{equation}
	Then, we have $J_{t,m}=\emptyset$ and $J_{t,m+1}=J_{t-1,m}\cup J_{t-1,m+1}\backslash (S_t\cap J_{t-1,m+1})$. Obviously, the result holds. 
\end{proof}

\textbf{Observation 3: } Assume $N_{u,\lceil\frac{m}{r}\rceil}\geq m$ for every $u$ in $\N$. 
In round $t=\lceil\frac{m+1}{r}\rceil$, we have:
$$ A_t =      J_{t-1,m}\cup J_{t-1,m+1}.$$
\begin{proof}
	Obviously, we have $N_{u,\lceil\frac{m+1}{r}\rceil-1}\geq m$ for any $u$ in $\N$ and thus, $J_{\lceil\frac{m+1}{r}\rceil-1,m}\cup J_{\lceil\frac{m+1}{r}\rceil-1,m+1}\cup \bar{J}_{\lceil\frac{m+1}{r}\rceil-1,m+1}=\N$.
	
	i) $\forall u\in J_{\lceil\frac{m+1}{r}\rceil-1,m}$, 
	
	$$rt-N_{u,t-1}=r\lceil\frac{m+1}{r}\rceil-m> 0$$
	
	$$\Rightarrow u\in A_t.$$
	
	ii) $\forall u\in J_{\lceil\frac{m+1}{r}\rceil-1,m+1}$, 
	
	$$rt-N_{u,t-1}=r\lceil\frac{m+1}{r}\rceil-(m+1)\geq 0$$
	
	$$\Rightarrow u\in A_t.$$
	
	iii) $\forall u\in \bar{J}_{\lceil\frac{m+1}{r}\rceil-1,m+1}$, $N_{u,\lceil\frac{m+1}{r}\rceil-1}\geq m+2$ and
	$$
	\begin{aligned}
	rt-N_{u,t-1}&= r\lceil\frac{m+1}{r}\rceil-(m+2)\\
	&\leq r\left(\lceil\frac{m+2}{r}\rceil-1\right)-(m+2)\\
	&< r\cdot \frac{m+2}{r}-(m+2)=0
	\end{aligned}
	$$
	
	$$\Rightarrow u\notin A_t.$$
\end{proof}

\begin{proof}[Proof of Lemma~\ref{lem:inductive_rslt}]
	We prove the lemma by induction.
	
	\underline{Induction base case ($m=0$):} At the very beginning, we have $N_{u,0}=0$ for any $u$ in $\N$ and thus, $\N = J_{0,0}\cup \bar{J}_{0,0}$ trivially. Moreover, we have 
	$\vert J_{0,0} \vert =n\leq \frac{1}{r} \cdot k$. That is, results $1.$ and $2.$ in Lemma~\ref{lem:inductive_rslt} hold for $m=0$ trivially.

	\underline{Inductive Step:} Assume that  both the results hold for $m$. Combining the above observations, we show the results hold for $m+1$ in the following.
	
	Consider three cases: 
	
	i) when $\lceil\frac{m+1}{r}\rceil>\lceil\frac{m}{r}\rceil+1$, there exists a round $\tau \in \{\lceil\frac{m}{r}\rceil+1, \cdots, \lceil\frac{m+1}{r}\rceil-1 \}
	$, such that $|A_{\tau}|<k$; 
	
	ii) when $\lceil\frac{m+1}{r}\rceil>\lceil\frac{m}{r}\rceil+1$, for every $\tau \in
	\{\lceil\frac{m}{r}\rceil+1, \cdots, \lceil\frac{m+1}{r}\rceil-1 \}
	$, $|A_{\tau}|\geq k$ holds; 
	
	iii) $\lceil\frac{m+1}{r}\rceil=\lceil\frac{m}{r}\rceil+1$.

	\textbf{Case i)}: There exists $\tau \in \mathcal{T}_{m+1}\backslash \{\lceil\frac{m+1}{r}\rceil\} $ such that $|A_{\tau}|<k$.
	
	Due to $A_{\tau} = J_{\tau-1,m}$ and $\N  = J_{\tau-1,m}\cup \bar{J}_{\tau-1,m}$, we have 
	$$\forall u \in J_{\tau-1,m} \Rightarrow u\in S_{\tau}\Rightarrow u\in J_{\tau,m+1}$$
	
	\begin{equation}
		\begin{aligned}
			\forall v \in \bar{J}_{\tau-1,m} 
			&\Rightarrow
			\left\{
			\begin{array}{ll}
				v\in \bar{J}_{\tau,m+1} &~\text{if}~ v\in S_{\tau}\\
				v\in \bar{J}_{\tau,m}
				&~\text{if}~ v\notin S_{\tau}
			\end{array}
			\right.
			\\
			&\Rightarrow v\in \bar{J}_{\tau,m}=J_{\tau, m+1} \cup \bar{J}_{\tau, m+1}
		\end{aligned}
	\end{equation}
	Hence, we derive $\N = J_{\tau, m+1} \cup \bar{J}_{\tau,m+1}$. For each element $u$, we have $N_{u,\lceil\frac{m+1}{r}\rceil}\geq N_{u,\tau}\geq m+1$, and thus, $u\in J_{\lceil\frac{m+1}{r}\rceil,m} \cup \bar{J}_{\lceil\frac{m+1}{r}\rceil,m+1}$. Therefore, we obtain the first result $\N = J_{\lceil\frac{m+1}{r}\rceil,m} \cup \bar{J}_{\lceil\frac{m+1}{r}\rceil,m+1}$ in Lemma~\ref{lem:inductive_rslt}.
	
	Besides, we have $J_{\lceil\frac{m+1}{r}\rceil-1,m}\subseteq J_{\tau,m}=\emptyset$ and then $A_{\lceil\frac{m+1}{r}\rceil} =J_{\lceil\frac{m+1}{r}\rceil-1,m+1}$ and $\N =J_{\lceil\frac{m+1}{r}\rceil-1,m+1}\cup \bar{J}_{\lceil\frac{m+1}{r}\rceil-1,m+1}$.

	1) If $|A_{\lceil\frac{m+1}{r}\rceil}|\leq k$, then all elements with $N_{u,\lceil\frac{m+1}{r}\rceil-1}=m+1$ are selected. We have 
	
	$\forall u \in J_{\lceil\frac{m+1}{r}\rceil-1, m+1} \Rightarrow u\in S_{\lceil\frac{m+1}{r}\rceil}$, and then, $u\in J_{\lceil\frac{m+1}{r}\rceil, m+2}$
	
	$\Rightarrow |J_{\lceil\frac{m+1}{r}\rceil,m+1}|=0$ satisfying the result trivially.
	
	2) If $|A_{\lceil\frac{m+1}{r}\rceil}|> k$, then $k$ elements from $J_{\lceil\frac{m+1}{r}\rceil-1,m+1}$ are selected in this round. Then, we have
	$$
	\begin{aligned}
	&|J_{\lceil\frac{m+1}{r}\rceil, m+1}|\\
	&= |J_{\lceil\frac{m+1}{r}\rceil-1, m+1}|- k\\
	&\leq n-k\\
	&\leq n-k+ \left(\frac{m+1}{r} +1 - \lceil\frac{m+1}{r}\rceil \right)k\\
	& \leq \frac{1}{r}\cdot k + \frac{m+1}{r}\cdot k- \lceil\frac{m+1}{r}\rceil \cdot k\\
	& \leq \left(\frac{m+2}{r}- \lceil\frac{m+1}{r}\rceil \right) k.
	\end{aligned}$$
	Hence, the second statement in Lemma~\ref{lem:inductive_rslt} holds for $m+1$.
	
	\textbf{Case ii):} When $\lceil\frac{m+1}{r}\rceil>\lceil\frac{m}{r}\rceil+1$, for every $\tau \in
	\{\lceil\frac{m}{r}\rceil+1, \cdots, \lceil\frac{m+1}{r}\rceil-1 \}
	$, $|A_{\tau}|\geq k$ holds.
	
	According to {Observation 2}, we have
	$$|J_{\lceil\frac{m}{r}\rceil+1,m}|=|J_{\lceil\frac{m}{r}\rceil,m}|-k$$
	$$|J_{\lceil\frac{m}{r}\rceil+2,m}|=|J_{\lceil\frac{m}{r}\rceil+1,m}|-k$$
	$$\cdots$$
	$$|J_{\lceil\frac{m+1}{r}\rceil-1,m}|=|J_{\lceil\frac{m+1}{r}\rceil-2,m}|-k$$

	\begin{equation}
	 \begin{aligned}
	\Rightarrow &|J_{\lceil\frac{m+1}{r}\rceil-1,m}| = |J_{\lceil\frac{m}{r}\rceil,m}|- \left(\lceil\frac{m+1}{r}\rceil-1-\lceil\frac{m}{r}\rceil \right)k\\
	& \leq \left(\frac{m+1}{r}- \lceil\frac{m}{r}\rceil \right) k- \left(\lceil\frac{m+1}{r}\rceil-1-\lceil\frac{m}{r}\rceil \right)k \\
	&= \left(\frac{m+1}{r}+1 - \lceil\frac{m+1}{r}\rceil \right)k. \label{eq:selection_no_update}
	\end{aligned}  
	\end{equation}

	In round $\lceil\frac{m+1}{r}\rceil$, we have $A_{\lceil\frac{m+1}{r}\rceil}=J_{\lceil\frac{m+1}{r}\rceil-1,m}\cup J_{\lceil\frac{m+1}{r}\rceil-1,m+1}$ according to Observation 3. 
	
	1) If $|A_{\lceil\frac{m+1}{r}\rceil}|\leq k$, 
	
	$\forall u \in J_{\lceil\frac{m+1}{r}\rceil-1, m}$ will be selected, i.e., $u\in S_{\lceil\frac{m+1}{r}\rceil}$, and then, $u\in J_{\lceil\frac{m+1}{r}\rceil, m+1}$;
	
	$\forall u \in J_{\lceil\frac{m+1}{r}\rceil-1, m+1}$ will be selected, i.e., $u\in S_{\lceil\frac{m+1}{r}\rceil}$, and then, $u\in J_{\lceil\frac{m+1}{r}\rceil, m+2}$. That is, the number of selection times for each element will be
	\begin{equation}
		\begin{aligned}
			&N_{u,\lceil\frac{m+1}{r}\rceil}&\\
			=&
			\left\{
			\begin{array}{ll}
				m+1, &  \forall u\in J_{\lceil\frac{m+1}{r}\rceil-1, m},\\
				m+2, &\forall u \in J_{\lceil\frac{m+1}{r}\rceil-1, m+1},\\
				N_{u,\lceil\frac{m+1}{r}\rceil-1}+1\geq m+3, &\forall u\in S_{\lceil\frac{m+1}{r}\rceil}\backslash A_{\lceil\frac{m+1}{r}\rceil},\\
				N_{u,\lceil\frac{m+1}{r}\rceil-1}\geq m+2, &\forall u\in \N\backslash S_{\lceil\frac{m+1}{r}\rceil}.
			\end{array}
			\right.
		\end{aligned}
	\end{equation}
	
	Then, we have $J_{\lceil\frac{m+1}{r}\rceil, m}=\emptyset$, indicating $\N=J_{\lceil\frac{m+1}{r}\rceil, m+1} \cup \bar{J}_{\lceil\frac{m+1}{r}\rceil, m+1}$.
	Moreover, by Eq.~\eqref{eq:selection_no_update} we have
	$$
	\begin{aligned}
	|J_{\lceil\frac{m+1}{r}\rceil, m+1}|&=|J_{\lceil\frac{m+1}{r}\rceil-1, m}| \\
	&\leq \left(\frac{m+1}{r}+1 - \lceil\frac{m+1}{r}\rceil \right)k\\
	& \leq \left(\frac{m+2}{r}- \lceil\frac{m+1}{r}\rceil \right) k,
	\end{aligned}$$
	where the last step is from $r\leq 1$. That is, the second result in Lemma~\ref{lem:inductive_rslt} holds for $m+1$ when $|A_{\lceil\frac{m+1}{r}\rceil}|\leq k$. 
	
	2) Consider the case when $|A_{\lceil\frac{m+1}{r}\rceil}|>k$. 
	
	Note that every element $u$ in $J_{\lceil\frac{m+1}{r}\rceil-1,m}$ with a larger fairness debt (smaller $N_{u,\lceil\frac{m+1}{r}\rceil-1}$) have a higher priority of being selected than any element in $J_{\lceil\frac{m+1}{r}\rceil-1, m+1}$.
	
	Since $|J_{\lceil\frac{m+1}{r}\rceil-1, m}| \leq \left(\frac{m+1}{r}+1 - \lceil\frac{m+1}{r}\rceil \right)k\leq k$, every element in $J_{\lceil\frac{m+1}{r}\rceil-1, m}$ will be selected. That is, for $\forall u \in J_{\lceil\frac{m+1}{r}\rceil-1, m}$,  we have $u\in S_{\lceil\frac{m+1}{r}\rceil}$, and then, $u\in J_{\lceil\frac{m+1}{r}\rceil, m+1}$. 
	
	Besides, $k-|J_{\lceil\frac{m+1}{r}\rceil-1,m}|$ elements from $J_{\lceil\frac{m+1}{r}\rceil-1,m+1}$ are selected.

	Then, we have $J_{\lceil\frac{m+1}{r}\rceil, m}=\emptyset$, indicating $\N={J}_{\lceil\frac{m+1}{r}\rceil, m+1} \cup \bar{J}_{\lceil\frac{m+1}{r}\rceil, m+1}$.
	
	For $\forall u \in S_{\lceil\frac{m+1}{r}\rceil}\backslash J_{\lceil\frac{m+1}{r}\rceil-1, m}$, we have $u\in J_{\lceil\frac{m+1}{r}\rceil-1, m+1}$ ($|A_{\lceil\frac{m+1}{r}\rceil}|>k$), and then, $u\in J_{\lceil\frac{m+1}{r}\rceil, m+2}$. In details, the number of selection times of each element will be
	\begin{equation}
		\begin{aligned}
			&N_{u,\lceil\frac{m+1}{r}\rceil}&\\
			=&
			\left\{
			\begin{array}{ll}
				m+1, &  \forall u\in J_{\lceil\frac{m+1}{r}\rceil-1, m},\\
				m+2, &\forall u \in J_{\lceil\frac{m+1}{r}\rceil-1, m+1} \cap S_{\lceil\frac{m+1}{r}\rceil},\\
				m+1, &\forall u\in
				J_{\lceil\frac{m+1}{r}\rceil-1, m+1} \backslash S_{\lceil\frac{m+1}{r}\rceil}, \\
				N_{u,\lceil\frac{m+1}{r}\rceil-1}\geq m+2, &\text{otherwise}
			\end{array}
			\right.
		\end{aligned}
	\end{equation}
	
	Finally, we have
	$$
	\begin{aligned}
	&|J_{\lceil\frac{m+1}{r}\rceil, m+1}|\\
	&=|J_{\lceil\frac{m+1}{r}\rceil-1, m}| + |J_{\lceil\frac{m+1}{r}\rceil-1, m+1}|- \left(k-|J_{\lceil\frac{m+1}{r}\rceil-1,m}|\right)\\
	&\leq n-k+|J_{\lceil\frac{m+1}{r}\rceil-1,m}|\\
	&\leq n-k+ \left(\frac{m+1}{r} - \lceil\frac{m}{r}\rceil \right)k\\
	& \leq \frac{1}{r}\cdot k -k+ \frac{m+1}{r}\cdot k- \lceil\frac{m}{r}\rceil \cdot k\\
	& = \frac{m+2}{r} \cdot k - (1+ \lceil\frac{m}{r}\rceil) k\\
	& \leq \left(\frac{m+2}{r}- \lceil\frac{m+1}{r}\rceil \right) k.
	\end{aligned}$$
    That is, the second result in Lemma~\ref{lem:inductive_rslt} holds for $m+1$ when $|A_{\lceil\frac{m+1}{r}\rceil}|>k$. 
    
\textbf{Case iii):} $\lceil\frac{m+1}{r}\rceil=\lceil\frac{m}{r}\rceil+1$.

In round $\lceil\frac{m+1}{r}\rceil$, we have $A_{\lceil\frac{m+1}{r}\rceil}=J_{\lceil\frac{m}{r}\rceil,m}\cup J_{\lceil\frac{m}{r}\rceil,m+1}$ according to Observation 3. 
We consider the two cases 1) $|A_{\lceil\frac{m+1}{r}\rceil}|\leq k$ and 2) $|A_{\lceil\frac{m+1}{r}\rceil}|>k$  as in the other two cases. 

1) If $|A_{\lceil\frac{m+1}{r}\rceil}|\leq k$, 
	
	$\forall u \in J_{\lceil\frac{m}{r}\rceil, m}$ will be selected, i.e., $u\in S_{\lceil\frac{m+1}{r}\rceil}$, and then, $u\in J_{\lceil\frac{m+1}{r}\rceil, m+1}$;
	
	$\forall u \in J_{\lceil\frac{m}{r}\rceil, m+1}$ will be selected, i.e., $u\in S_{\lceil\frac{m+1}{r}\rceil}$, and then, $u\in J_{\lceil\frac{m+1}{r}\rceil, m+2}$. That is, the number of selection times for each element will be
	\begin{equation}
		\begin{aligned}
			&N_{u,\lceil\frac{m+1}{r}\rceil}&\\
			=&
			\left\{
			\begin{array}{ll}
				m+1, &  \forall u\in J_{\lceil\frac{m}{r}\rceil, m},\\
				m+2, &\forall u \in J_{\lceil\frac{m}{r}\rceil, m+1},\\
				N_{u,\lceil\frac{m}{r}\rceil}+1\geq m+3, &\forall u\in S_{\lceil\frac{m+1}{r}\rceil}\backslash A_{\lceil\frac{m+1}{r}\rceil},\\
				N_{u,\lceil\frac{m}{r}\rceil}\geq m+2, &\forall u\in \N\backslash S_{\lceil\frac{m+1}{r}\rceil}.
			\end{array}
			\right.
		\end{aligned}
	\end{equation}
	
	Then, we have $J_{\lceil\frac{m+1}{r}\rceil, m}=\emptyset$, indicating $\N={J}_{\lceil\frac{m+1}{r}\rceil, m+1} \cup \bar{J}_{\lceil\frac{m+1}{r}\rceil, m+1}$.
	Moreover, we have
	$$
	\begin{aligned}
	|J_{\lceil\frac{m+1}{r}\rceil, m+1}|&=|J_{\lceil\frac{m}{r}\rceil, m}| \\
	& \leq \left(\frac{m+1}{r}- \lceil\frac{m}{r}\rceil \right) k \\
	&= \left(\frac{m+1}{r}+1 - \lceil\frac{m+1}{r}\rceil \right)k\\
	& \leq \left(\frac{m+2}{r}- \lceil\frac{m+1}{r}\rceil \right) k,
	\end{aligned}$$
	where the last step is from $r\leq 1$. That is, the second result in Lemma~\ref{lem:inductive_rslt} holds for $m+1$ when $|A_{\lceil\frac{m+1}{r}\rceil}|< k$.
	
	2) Consider the case when $|A_{\lceil\frac{m+1}{r}\rceil}|>k$. 
	
	Note that every element $u$ in $J_{\lceil\frac{m}{r}\rceil,m}$ with a larger fairness debt (smaller $N_{u,\lceil\frac{m}{r}\rceil}$) have a higher priority of being selected than any element in $J_{\lceil\frac{m}{r}\rceil, m+1}$.
	
	Since $|J_{\lceil\frac{m}{r}\rceil, m}| \leq \left(\frac{m+1}{r}+1 - \lceil\frac{m+1}{r}\rceil \right)k\leq k$, every element in $J_{\lceil\frac{m}{r}\rceil, m}$ will be selected. That is, for $\forall u \in J_{\lceil\frac{m}{r}\rceil, m}$,  we have $u\in S_{\lceil\frac{m+1}{r}\rceil}$, and then, $u\in J_{\lceil\frac{m+1}{r}\rceil, m+1}$. 
	
	Besides, $k-|J_{\lceil\frac{m}{r}\rceil,m}|$ elements from $J_{\lceil\frac{m}{r}\rceil,m+1}$ are selected.

	Then, we have $J_{\lceil\frac{m+1}{r}\rceil, m}=\emptyset$, indicating $\N={J}_{\lceil\frac{m+1}{r}\rceil, m+1} \cup \bar{J}_{\lceil\frac{m+1}{r}\rceil, m+1}$.
	
	For $\forall u \in S_{\lceil\frac{m+1}{r}\rceil}\backslash J_{\lceil\frac{m}{r}\rceil, m}$, we have $u\in J_{\lceil\frac{m}{r}\rceil, m+1}$ ($|A_{\lceil\frac{m+1}{r}\rceil}|>k$), and then, $u\in J_{\lceil\frac{m+1}{r}\rceil, m+2}$. In details, the number of selection times of each element will be
	\begin{equation}
		\begin{aligned}
			&N_{u,\lceil\frac{m+1}{r}\rceil}&\\
			=&
			\left\{
			\begin{array}{ll}
				m+1, &  \forall u\in J_{\lceil\frac{m}{r}\rceil, m},\\
				m+2, &\forall u \in J_{\lceil\frac{m}{r}\rceil, m+1} \cap S_{\lceil\frac{m+1}{r}\rceil},\\
				m+1, &\forall u\in
				J_{\lceil\frac{m}{r}\rceil, m+1} \backslash S_{\lceil\frac{m+1}{r}\rceil}, \\
				N_{u,\lceil\frac{m}{r}\rceil}\geq m+2, &\forall u\in \N\backslash J_{\lceil\frac{m}{r}\rceil, m}\backslash J_{\lceil\frac{m}{r}\rceil, m+1}.
			\end{array}
			\right.
		\end{aligned}
	\end{equation}
	
	Finally, we have
	$$
	\begin{aligned}
	&|J_{\lceil\frac{m+1}{r}\rceil, m+1}|\\
	&=|J_{\lceil\frac{m+1}{r}\rceil-1, m}| + |J_{\lceil\frac{m+1}{r}\rceil-1, m}|- \left(k-|J_{\lceil\frac{m+1}{r}\rceil-1,m}|\right)\\
	&\leq n-k+|J_{\lceil\frac{m+1}{r}\rceil-1,m}|\\
	&\leq n-k+ \left(\frac{m+1}{r} - \lceil\frac{m}{r}\rceil \right)k\\
	& \leq \frac{1}{r}\cdot k -k+ \frac{m+1}{r}\cdot k- \lceil\frac{m}{r}\rceil \cdot k\\
	& = \frac{m+2}{r} \cdot k - (1+ \lceil\frac{m}{r}\rceil) k\\
	& \leq \left(\frac{m+2}{r}- \lceil\frac{m+1}{r}\rceil \right) k.
	\end{aligned}$$
	That is, the second result in Lemma~\ref{lem:inductive_rslt} holds for $m+1$ when $|A_{\lceil\frac{m+1}{r}\rceil}|\leq k$.
Therefore, the two results in Lemma~\ref{lem:inductive_rslt} hold for any $m\geq 0$. We complete the proof. 
\end{proof}
\end{document}